\documentclass{amsart}

\usepackage[lite]{amsrefs}
\usepackage{amssymb}
\usepackage[all,cmtip]{xy}
\usepackage{amsthm}
\usepackage{amsmath}
\usepackage{amsfonts}
\usepackage[hidelinks]{hyperref}
\usepackage{enumerate}
\usepackage{physics}
\usepackage[titletoc]{appendix}
\usepackage{fancyhdr}
\numberwithin{equation}{section}
\newtheorem{definition}{Definition}[section]
\newtheorem{example}[definition]{Example}
\newtheorem{theorem}[definition]{Theorem}
\newtheorem{lemma}[definition]{Lemma}

\newtheorem{proposition}[definition]{Proposition}
\newtheorem{remark}[definition]{Remark}

\newtheorem{ansatz}[definition]{Ansatz}

\begin{document}
\title[Superconformal Structures of SUSY Free Fermion Vertex Algebras]{Superconformal Structures of Supersymmetric Free Fermion Vertex Algebras}

\author{Sangwon Yoon}
\address{Department of Mathematical Sciences, Seoul National University, Gwanak-ro 1, Gwanak-gu, Seoul 08826, Korea}
\email{ysw317@snu.ac.kr}

\begin{abstract}
In this paper, we define the shifted superconformal vector of supersymmetric charged free fermion vertex algebras, which is a 1-parameter deformation of the superconformal vector of the SUSY $bc$-$\beta\gamma$ system. Moreover, we find the corresponding shifted $N=2$ superconformal symmetry of SUSY charged free fermion vertex algebras, by using the $N_{K}=1$ SUSY vertex algebra formalism. Finally, in order to describe the shifted $N=2$ superconformal symmetry of the SUSY charged free fermion vertex algebra by $N=2$ superfields, we construct an $N_{K}=2$ SUSY version of the $bc$-$\beta\gamma$ system.
\end{abstract}

\maketitle

\section{Introduction}\label{section1}
\setcounter{equation}{0}

Superconformal symmetries have been studied widely in both mathematics and physics literature, see for example \cites{Ada99, GKO86, KW04, SS87}. Especially in the theory of vertex algebras, superconformal vectors give rise to the superconformal structures of vertex algebras (see \cite{Kac98}).
\\
\indent
In this paper, we study the superconformal structure and the deformation of superconformal vectors of the supersymmetric charged free fermion vertex algebras. In \cites{DSK06, KRW03}, the non-SUSY charged free fermion vertex algebra appears as a ghost part of the non-SUSY $W$-algebra. Just as in the non-SUSY case, the SUSY charged free fermion vertex algebra is the ghost part of the SUSY $W$-algebra, which was studied in \cites{MR94, MRS21}. On the other hand, conformal vectors play an important role in the theory of vertex algebras, for example, they induce the notion of energy-momentum fields (see \cite{Kac98}). Also, superconformal vectors can be regarded as supersymmetric analogues of conformal vectors.
\\
\indent
Despite some difficulties in finding superconformal vectors in general, the constructions of superconformal vectors were studied in some cases. For example, the superconformal vector of the non-SUSY free fermion vertex algebra was studied in \cite{GKO86}, and the Kac-Todorov construction (see \cite{KT85}) of the non-SUSY affine vertex algebras was given in \cite{Kac98}. In the SUSY case, some basics about the superconformal structures of SUSY vertex algebras and especially superconformal vectors of the SUSY affine vertex algebras and the SUSY $bc$-${\beta}{\gamma}$ system can be found in \cite{HK07}.
\\
\indent
In the context of BRST cohomology, deformations of the (super)conformal vectors are also important, in addition to the construction of (super)conformal vectors.
In \cites{EFH98}, the 1-parameter deformation of the conformal vector of the non-SUSY $bc$-${\beta}{\gamma}$ system was introduced. Also, the modification of the Sugawara construction of the conformal vector of the non-SUSY affine vertex algebra was used in \cites{DSK06, KRW03}, and for the SUSY counterpart, the modified Kac-Todorov superconformal vector of the SUSY affine vertex algebra was presented in \cite{HZ10}.
\\
\indent
In the present paper, we give a 1-parameter deformation of the superconformal vector of the SUSY $bc$-${\beta}{\gamma}$ system (Theorem \ref{theorem4.6}), which is called the shifted superconformal vector of the SUSY charged free fermion vertex algebras, in order to obtain varying conformal weights of monomial fields of the SUSY charged free fermion vertex algebras. Also we find an $N=2$ superconformal symmetry (Theorem \ref{theorem4.14}), compatible with our shifted superconformal vector.
\\
\indent
The organization of this paper is as follows. In Section \ref{section2}, following \cites{HK07, Kac98}, we briefly review the definitions and basic properties of the vertex algebras and the SUSY vertex algebras. In Section \ref{section3}, we recall the definitions of superconformal vertex algebras in the context of the vertex algebras and the SUSY vertex algebras. In Section \ref{section4}, we state the main results of this work. We define the shifted superconformal vector of the supersymmetric charged free fermion vertex algebras and find the corresponding $N=2$ superconformal symmetry. At the end of this section, we describe the shifted $N=2$ superconformal structure, using the $N_{K}=2$ SUSY vertex algebra formalism. In Appendix \ref{sectionA}, we present some calculations of lambda brackets of the SUSY charged free fermion vertex algebras. Appendix \ref{sectionB} contains the charge decomposition of the SUSY charged free fermion vertex algebras.

\section{Vertex algebras and SUSY vertex algebras}\label{section2}
\setcounter{equation}{0}

In this section, we recollect the definitions of vertex algebras and SUSY vertex algebras. We only consider the $N_{K}=1$ and $N_{K}=2$ SUSY vertex algebras. The main references are \cite{DSK06}, \cite{HK07} and \cite{Kac98}.

\begin{definition}\label{definition2.1}{\rm \cite{Kac98}}
\rm Let $V$ be a vector superspace. Let $z$ be an even indeterminate. A \textit{field} is an End($V$)-valued formal distribution of the form
\begin{align*}
\begin{split}
a(z) = \sum\limits_{j \in \mathbb{Z}}{z}^{-j-1}{a}_{(j)},
\end{split}
\end{align*}
where ${a}_{(j)}v = 0$ for all but finitely many $j \in \mathbb{Z}$, for each $v \in V$.
\end{definition}

\begin{definition}\label{definition2.2}{\rm \cites{DSK06, Kac98}}
\rm A \textit{vertex algebra} $(V, \ket{0}, \partial, Y)$ consists of a vector superspace $V$ over $\mathbb{C}$, an even vector $\ket{0}$, an even endomorphism $\partial$, and the \textit{state-field correspondence} $Y$ which is a parity preserving linear map from $V$ to the space of End($V$)-valued fields, such that the following axioms hold:
\\
$\bullet$ (Vacuum)
\begin{align*}
\begin{split}
\partial\ket{0}=0, \quad Y(a, z)\ket{0}|_{z=0} = a,
\end{split}
\end{align*}
$\bullet$ (Translation covariance)
\begin{align*}
\begin{split}
[\partial, Y(a, z)] = {\partial}_{z}Y(a, z),
\end{split}
\end{align*}
$\bullet$ (Locality) for any $a$, $b$ $\in$ $V$, there exists $n \in \mathbb{Z}_{\geq 0}$ such that
\begin{align*}
\begin{split}
(z-w)^{n}[Y(a, z), Y(b, w)] = 0.
\end{split}
\end{align*}
The \textit{normally ordered product} on $V$ is defined by
\begin{align*}
\begin{split}
:ab: \ =  {a}_{(-1)}b.
\end{split}
\end{align*}
In this paper, we just denote by $ab$ the normally ordered product and $:a:bc::$ is denoted by $abc$.
\end{definition}

\begin{definition}\label{definition2.3}{\rm \cite{DSK06}}
\rm A \textit{Lie conformal algebra} is a $\mathbb{Z}/2\mathbb{Z}$-graded $\mathbb{C}[\partial]$-module $R$ with a $\lambda$-\textit{bracket} which is a parity preserving $\mathbb{C}$-bilinear map:
\begin{align*}
\begin{split}
[ \ _{\lambda} \ ] : R \otimes R \rightarrow \mathbb{C}[\lambda] \otimes R, \quad a \otimes b \mapsto [a \ _{\lambda} \ b],
\end{split}
\end{align*}
satisfying the following conditions:
\\
$\bullet$ (Sesquilinearity)
\begin{align*}
\begin{split}
[\partial a \ _{\lambda} \ b] = - \lambda[a \ _{\lambda} \ b], \quad [a \ _{\lambda} \ \partial b] = (\partial + \lambda)[a \ _{\lambda} \ b],
\end{split}
\end{align*}
$\bullet$ (Skew-symmetry)
\begin{align*}
\begin{split}
[b \ _{\lambda} \ a] = (-1)^{ab+1}[a \ _{-\partial-\lambda} \ b],
\end{split}
\end{align*}
$\bullet$ (Jacobi identity)
\begin{align*}
\begin{split}
[a \ _{\lambda} \ [b \ _{\mu} \ c]] = [[a \ _{\lambda} \ b] \ _{\lambda+\mu} \ c] + (-1)^{ab}[b \ _{\mu} \ [a \ _{\lambda} \ c]],
\end{split}
\end{align*}
where we write the $\lambda$-brackets as:
\begin{align*}
\begin{split}
[a \ _{\lambda} \ b] = \sum\limits_{j \in \mathbb{Z}_{\geq 0}}\frac{{\lambda}^{j}}{j!}a_{(j)}b.
\end{split}
\end{align*}
\end{definition}

\begin{definition}\label{definition2.4}{\rm \cite{DSK06}}
\rm A \textit{vertex algebra} is a tuple $(V, \partial, [ \ _{\lambda} \ ], \ket{0}, : \ :)$ such that:
\\
$\bullet$ $(V, \partial, [ \ _{\lambda} \ ])$ is a Lie conformal algebra,
\\
$\bullet$ $(V, \partial, \ket{0},: \:)$ is a unital differential superalgebra with a derivation $\partial$, satisfying the following properties:
\\
(Quasi-commutativity)
\begin{align*}
\begin{split}
ab - (-1)^{ab}ba = \int_{-\partial}^{0} [a \ _{\lambda} \ b] \ d\lambda,
\end{split}
\end{align*}
(Quasi-associativity)
\begin{align*}
\begin{split}
(ab)c - a(bc) = \left( \int_{0}^{\partial} \ d\lambda a \right)[b \ _{\lambda} \ c] + (-1)^{ab}\left( \int_{0}^{\partial} \ d\lambda b \right)[a \ _{\lambda} \ c],
\end{split}
\end{align*}
$\bullet$ the $\lambda$-bracket and the product $: \:$ are related by:
\\
(Non-commutative Wick formula)
\begin{align*}
\begin{split}
[a \ _{\lambda} \ bc] = [a \ _{\lambda} \ b]c + (-1)^{ab}b[a \ _{\lambda} \ c] + \int_{0}^{\lambda} [[a \ _{\lambda} \ b] \ _{\mu} \ c] \ d\mu.
\end{split}
\end{align*}
\end{definition}

\begin{example}\label{exmaple2.5}{\rm \cites{DSK06, KRW03}}
\rm Let $U$ be a finite-dimensional vector superspace and $A$ be a basis of $U$. Define two vector superspaces:
\begin{align*}
\begin{split}
\varphi_{U} \simeq {\Pi}U, \quad \varphi^{U} \simeq {\Pi}{U}^{*},
\end{split}
\end{align*}
whose basis elements are denoted by $\varphi_{a}$ and $\varphi^{a}$ respectively. Here ${\Pi}$ denotes the parity reversing functor and ${U}^{*}$ is the dual vector space of $U$. Define the $\lambda$-brackets on $R^{ch} = \mathbb{C}[\partial] \otimes (\varphi_{U} \oplus \varphi^{U})$ as:
\begin{align*}
\begin{split}
[\varphi_{a} \ _{\lambda} \ \varphi^{b}] = {\delta}_{ab}, \quad [\varphi_{a} \ _{\lambda} \ \varphi_{b}] = [\varphi^{a} \ _{\lambda} \ \varphi^{b}] = 0,
\end{split}
\end{align*}
for $a, b \in A$. Then $R^{ch}$ is a Lie conformal algebra, and the \textit{charged free fermion vertex algebra} $F^{ch}$ associated to the vector superspace $U$ is the universal enveloping vertex algebra $V(R^{ch})$.
\end{example}

\begin{example}\label{exmaple2.6}{\rm \cite{EHKZ09}}
\rm Let $U$ be the vector superspace with the basis $A = A_{\bar{0}} \cup A_{\bar{1}}$ be a basis of $U$, where $A_{\bar{0}} = \{\beta, \gamma\}$ is the even part and $A_{\bar{1}} = \{b, c\}$ is the odd part of the basis.
Define the $\lambda$-brackets on $R^{bc\beta\gamma} = \mathbb{C}[\partial] \otimes U$ as:
\begin{align*}
\begin{split}
[\beta \ _{\lambda} \ \gamma] = 1, &\quad [\gamma \ _{\lambda} \ \beta] = -1,
\\
[b \ _{\lambda} \ c] = 1, &\quad [c \ _{\lambda} \ b] = 1,
\end{split}
\end{align*}
and other $\lambda$-brackets on the basis elements are zeros. Then $R^{bc\beta\gamma}$ is a Lie conformal algebra, and the \textit{$bc$-$\beta\gamma$ system} is the universal enveloping vertex algebra $V(R^{bc\beta\gamma})$.
\end{example}

\begin{definition}\label{definition2.7}{\rm \cite{HK07}}
\rm Let $V$ be a vector superspace. Let $z$ be an even indeterminate and $\theta$ be an odd indeterminate such that ${\theta}^{2} = 0$ and $\theta z = z\theta$. An $N = 1$ \textit{superfield} is an End($V$)-valued formal distribution of the form
\begin{align*}
\begin{split}
a(z, \theta) = \sum\limits_{j \in \mathbb{Z}}{z}^{-j-1}{a}_{(j|1)} + \theta\sum\limits_{j \in \mathbb{Z}}{z}^{-j-1}{a}_{(j|0)},
\end{split}
\end{align*}
where ${a}_{(j|*)}v = 0$ for all but finitely many $j \in \mathbb{Z}$, for each $v \in V$.
\end{definition}

\begin{definition}\label{definition2.8}{\rm \cite{HK07}}
\rm An $N_{K} = 1$ \textit{SUSY vertex algebra} is a tuple $(V, \ket{0}, D, Y)$ consisting of a vector superspace $V$ over $\mathbb{C}$, an even vector $\ket{0}$, an odd endomorphism $D$, and the \textit{state-field correspondence} $Y$ which is a parity preserving linear map from $V$ to the space of End($V$)-valued $N = 1$ superfields, such that the following axioms hold:
\\
$\bullet$ (Vacuum)
\begin{align*}
\begin{split}
D\ket{0}=0, \quad Y(a, z, \theta)\ket{0}|_{z=0, \theta=0} = a,
\end{split}
\end{align*}
$\bullet$ (Translation covariance)
\begin{align*}
\begin{split}
[D, Y(a, z, \theta)] = ({\partial}_{\theta} - {\theta}{\partial}_{z})Y(a, z, \theta),
\end{split}
\end{align*}
$\bullet$ (Locality) for any $a$, $b$ $\in$ $V$, there exists $n \in \mathbb{Z}_{\geq 0}$ such that
\begin{align*}
\begin{split}
(z-w)^{n}[Y(a, z, \theta), Y(b, w, \zeta)] = 0.
\end{split}
\end{align*}
The \textit{normally ordered product} on $V$ is defined by
\begin{align*}
\begin{split}
:ab: \ =  {a}_{(-1|1)}b.
\end{split}
\end{align*}
\end{definition}

\begin{remark}\label{remark2.9}{\rm \cite{HK07}}
\rm Consider the noncommutative associative superalgebra $\mathcal{H}_{N=1} = \mathbb{C}[\partial, D]$ generated by an even generator $\partial$ and an odd generator $D$, with relations:
\begin{align*}
\begin{split}
D^{2} = \partial, \quad [\partial, D] = 0, 
\end{split}
\end{align*}
and consider the noncommutative associative superalgebra 
$\mathcal{L}_{N=1} = \mathbb{C}[\lambda, \chi]$ generated by an even generator $\lambda$ and an odd generator $\chi$, with relations:
\begin{align*}
\begin{split}
{\chi}^{2} = -\lambda, \quad [\lambda, \chi] = 0.
\end{split}
\end{align*}
\end{remark}

\begin{definition}\label{definition2.10}{\rm \cites{HK07, MRS21}}
\rm An $N_{K} = 1$ \textit{supersymmetric Lie conformal algebra} is a $\mathbb{Z}/2\mathbb{Z}$-graded $\mathcal{H}_{N=1}$-module $R$ with a $\Lambda$-\textit{bracket} which is a $\mathbb{C}$-bilinear map of degree 1:
\begin{align*}
\begin{split}
[ \ _{\Lambda} \ ] : R \otimes R \rightarrow \mathcal{L}_{N=1} \otimes R, \quad a \otimes b \mapsto [a \ _{\Lambda} \ b],
\end{split}
\end{align*}
satisfying the following conditions:
\\
$\bullet$ (Sesquilinearity)
\begin{align*}
\begin{split}
[D a \ _{\Lambda} \ b] = \chi[a \ _{\Lambda} \ b], \quad [a \ _{\Lambda} \ Db] = (-1)^{a+1}(D + \chi)[a \ _{\Lambda} \ b],
\end{split}
\end{align*}
where $D$ and $\chi$ are subject to the relation $[D, \chi] = 2\lambda$,
\\
$\bullet$ (Skew-symmetry)
\begin{align*}
\begin{split}
[b \ _{\Lambda} \ a] = (-1)^{ab}[a \ _{-\nabla-\Lambda} \ b],
\end{split}
\end{align*}
where $-\nabla-\Lambda = (-\partial-\lambda, -D-\chi)$ with the relations
\begin{align*}
\begin{split}
[D, \chi] = 2\lambda, \quad [\partial, \lambda] = [\partial, \chi] = [D, \lambda] = 0,
\end{split}
\end{align*}
$\bullet$ (Jacobi identity)
\begin{align*}
\begin{split}
[a \ _{\Lambda} \ [b \ _{\Gamma} \ c]] = (-1)^{a+1}[[a \ _{\Lambda} \ b] \ _{\Lambda+\Gamma} \ c] + (-1)^{(a+1)(b+1)}[b \ _{\Gamma} \ [a \ _{\Lambda} \ c]],
\end{split}
\end{align*}
where $\Gamma = (\gamma, \eta)$ and $\Lambda+\Gamma = (\lambda+\gamma, \chi+\eta)$ with relations:
\begin{align*}
\begin{split}
[\eta, \eta] = -2\gamma, \quad [\gamma, \gamma] = [\gamma, \eta] = [\lambda, \gamma] = [\lambda, \eta] = [\chi, \gamma] = [\chi, \eta] = 0.
\end{split}
\end{align*}
We write the $\Lambda$-brackets as:
\begin{align*}
\begin{split}
[a \ _{\Lambda} \ b] = \sum\limits_{j \in \mathbb{Z}_{\geq 0}}\frac{{\lambda}^{j}}{j!}a_{(j|0)}b + \chi\sum\limits_{j \in \mathbb{Z}_{\geq 0}}\frac{{\lambda}^{j}}{j!}a_{(j|1)}b.
\end{split}
\end{align*}
\end{definition}

\begin{definition}\label{definition2.11}{\rm \cite{HK07}}
\rm An $N_{K} = 1$ \textit{supersymmetric vertex algebra} is a tuple $($$V$, $D$, $[$ $_{\Lambda}$ $]$, $\ket{0}$, : :$)$ such that:
\\
$\bullet$ $(V, D, [ \ _{\Lambda} \ ])$ is an $N_{K}=1$ SUSY Lie conformal algebra,
\\
$\bullet$ $(V, D, \ket{0},: \ :)$ is a unital differential superalgebra with an odd derivation $D$, satisfying the following properties:
\\
(Quasi-commutativity)
\begin{align*}
\begin{split}
ab - (-1)^{ab}ba = \int_{-\nabla}^{0} [a \ _{\Lambda} \ b] \ d\Lambda,
\end{split}
\end{align*}
(Quasi-associativity)
\begin{align*}
\begin{split}
(ab)c - a(bc) = \left( \int_{0}^{\nabla} \ d\Lambda a \right)[b \ _{\Lambda} \ c] + (-1)^{ab}\left( \int_{0}^{\nabla} \ d\Lambda b \right)[a \ _{\Lambda} \ c],
\end{split}
\end{align*}
$\bullet$ the $\Lambda$-bracket and the product $: \:$ are related by:
\\
(Non-commutative Wick formula)
\begin{align*}
\begin{split}
[a \ _{\Lambda} \ bc] = [a \ _{\Lambda} \ b]c + (-1)^{(a+1)b}b[a \ _{\Lambda} \ c] + \int_{0}^{\Lambda} [[a \ _{\Lambda} \ b] \ _{\Gamma} \ c] \ d\Gamma,
\end{split}
\end{align*}
where the integral $\int_{0}^{\Lambda} d\Gamma$ is computed as ${\partial}_{\eta}\int_{0}^{\lambda} d\gamma$.
\end{definition}

\begin{example}\label{exmaple2.12}{\rm \cite{MRS21}}
\rm Let $U$ be a finite-dimensional vector superspace and $A = A_{\bar{0}} \cup A_{\bar{1}}$ be a basis of $U$, where $A_{\bar{0}}$ is the even part and $A_{\bar{1}}$ is the odd part of the basis. Define two vector superspaces:
\begin{align*}
\begin{split}
\phi_{U} \simeq U, \quad \phi^{\bar{U}} \simeq \Pi{U}^{*},
\end{split}
\end{align*}
whose basis elements are denoted by $\phi_{a}$ and $\phi^{\bar{a}}$ respectively. Define the $\Lambda$-brackets on $R^{ch}_{N=1} = \mathcal{H}_{N=1} \otimes (\phi_{U} \oplus \phi^{\bar{U}})$ as:
\begin{align*}
\begin{split}
[\phi_{a} \ _{\Lambda} \ \phi^{\bar{b}}] = {\delta}_{ab}, \quad [\phi_{a} \ _{\Lambda} \ \phi_{b}] = [\phi^{\bar{a}} \ _{\Lambda} \ \phi^{\bar{b}}] = 0,
\end{split}
\end{align*}
for $a, b \in A$. Then $R^{ch}_{N=1}$ is an $N_{K} = 1$ supersymmetric Lie conformal algebra, and the \textit{supersymmetric charged free fermion vertex algebra} $F^{ch}_{N=1}$ associated to the vector superspace $U$ is the universal enveloping supersymmetric vertex algebra $V(R^{ch}_{N=1})$. For the description of the universal enveloping supersymmetric vertex algebras, see Definition 2.13 of \cite{MRS21}.
\end{example}

\begin{definition}\label{definition2.13}{\rm \cite{HK07}}
\rm Let $V$ be a vector superspace. Let $z$ be an even indeterminate and ${\theta}^{1}$, ${\theta}^{2}$ be odd indeterminates such that $[{\theta}^{i}, {\theta}^{j}] = 0$ and $[{\theta}^{i}, z] = 0$. An $N = 2$ \textit{superfield} is an End($V$)-valued formal distribution of the form
\begin{align*}
\begin{split}
a(z, {\theta}^{1}, {\theta}^{2}) = \sum\limits_{j \in \mathbb{Z}}{z}^{-j-1}{a}&_{(j|11)} + {\theta}^{1}\sum\limits_{j \in \mathbb{Z}}{z}^{-j-1}{a}_{(j|01)}
\\
&+ {\theta}^{2}\sum\limits_{j \in \mathbb{Z}}{z}^{-j-1}{a}_{(j|10)} + {\theta}^{1}{\theta}^{2}\sum\limits_{j \in \mathbb{Z}}{z}^{-j-1}{a}_{(j|00)},
\end{split}
\end{align*}
where ${a}_{(j|**)}v = 0$ for all but finitely many $j \in \mathbb{Z}$, for each $v \in V$.
\end{definition}

\begin{definition}\label{definition2.14}{\rm \cite{HK07}}
\rm An $N_{K} = 2$ \textit{SUSY vertex algebra} is a tuple $(V, \ket{0}, D^{1}, D^{2}, Y)$ consisting of a vector superspace $V$ over $\mathbb{C}$, an even vector $\ket{0}$, an odd endomorphisms $D^{1}$, $D^{2}$ and the \textit{state-field correspondence} $Y$ which is a parity preserving linear map from $V$ to the space of End($V$)-valued $N = 2$ superfields, such that the following axioms hold:
\\
$\bullet$ (Vacuum)
\begin{align*}
\begin{split}
{D}^{i}\ket{0}=0, \quad Y(a, z, {\theta}^{1}, {\theta}^{2})\ket{0}|_{z=0, {\theta}^{i}=0} = a,
\end{split}
\end{align*}
$\bullet$ (Translation covariance)
\begin{align*}
\begin{split}
[D^{i}, Y(a, z, {\theta}^{1}, {\theta}^{2})] = ({\partial}_{{\theta}^{i}} - {{\theta}^{i}}{\partial}_{z})Y(a, z, {\theta}^{1}, {\theta}^{2}),
\end{split}
\end{align*}
$\bullet$ (Locality) for any $a$, $b$ $\in$ $V$, there exists $n \in \mathbb{Z}_{\geq 0}$ such that
\begin{align*}
\begin{split}
(z-w)^{n}[Y(a, z, {\theta}^{1}, {\theta}^{2}), Y(b, w, {\zeta}^{1}, {\zeta}^{2})] = 0.
\end{split}
\end{align*}
The \textit{normally ordered product} on $V$ is defined by
\begin{align*}
\begin{split}
:ab: \ =  {a}_{(-1|11)}b.
\end{split}
\end{align*}
\end{definition}

\begin{remark}\label{remark2.15}{\rm \cite{HK07}}
\rm Consider the noncommutative associative superalgebra ${\mathcal{H}}_{N=2}$ $=$ $\mathbb{C}[\partial, D^{1}, D^{2}]$ generated by an even generator $\partial$ and odd generators $D^{1}$, $D^{2}$ with relations:
\begin{align*}
\begin{split}
[D^{i}, D^{j}] = 2{\delta}_{ij}\partial, \quad [\partial, D^{i}] = 0, 
\end{split}
\end{align*}
and the noncommutative associative superalgebra 
${\mathcal{L}}_{N=2} = \mathbb{C}[\lambda, {\chi}^{1}, {\chi}^{2}]$ generated by an even generator $\lambda$ and odd generators ${\chi}^{1}$, ${\chi}^{2}$ with relations:
\begin{align*}
\begin{split}
[{\chi}^{i}, {\chi}^{j}] = -2{\delta}_{ij}\lambda, \quad [\lambda, {\chi}^{i}] = 0.
\end{split}
\end{align*}
\end{remark}

\begin{definition}\label{definition2.16}{\rm \cites{EHZ12, HK07}}
\rm An $N_{K} = 2$ \textit{supersymmetric Lie conformal algebra} is a $\mathbb{Z}/2\mathbb{Z}$-graded ${\mathcal{H}}_{N=2}$-module $R$ with a $\Lambda$-\textit{bracket} which is a parity preserving $\mathbb{C}$-bilinear map:
\begin{align*}
\begin{split}
[ \ _{\Lambda} \ ] : R \otimes R \rightarrow {\mathcal{L}}_{N=2} \otimes R, \quad a \otimes b \mapsto [a \ _{\Lambda} \ b],
\end{split}
\end{align*}
satisfying the following conditions:
\\
$\bullet$ (Sesquilinearity)
\begin{align*}
\begin{split}
[D^{i} a \ _{\Lambda} \ b] = -{\chi}^{i}[a \ _{\Lambda} \ b], \quad [a \ _{\Lambda} \ D^{i}b] = (-1)^{a}(D^{i} + {\chi}^{i})[a \ _{\Lambda} \ b],
\end{split}
\end{align*}
where $D^{i}$ and ${\chi}^{i}$ are subject to the relation $[D^{i}, {\chi}^{j}] = 2{\delta}_{ij}\lambda$,
\\
$\bullet$ (Skew-symmetry)
\begin{align*}
\begin{split}
[b \ _{\Lambda} \ a] = (-1)^{ab+1}[a \ _{-\nabla-\Lambda} \ b],
\end{split}
\end{align*}
where $-\nabla-\Lambda = (- \partial - \lambda, - D^{1} - {\chi}^{1}, - D^{2} - {\chi}^{2})$ with the relations
\begin{align*}
\begin{split}
[D^{i}, {\chi}^{j}] = 2{\delta}_{ij}\lambda, \quad [\partial, \lambda] = [\partial, {\chi}^{i}] = [D^{i}, \lambda] = 0,
\end{split}
\end{align*}
$\bullet$ (Jacobi identity)
\begin{align*}
\begin{split}
[a \ _{\Lambda} \ [b \ _{\Gamma} \ c]] = [[a \ _{\Lambda} \ b] \ _{\Lambda+\Gamma} \ c] + (-1)^{ab}[b \ _{\Gamma} \ [a \ _{\Lambda} \ c]],
\end{split}
\end{align*}
where $\Gamma = (\gamma, {\eta}^{1}, {\eta}^{2})$ and $\Lambda+\Gamma = (\lambda+\gamma, {\chi}^{1}+{\eta}^{1}, {\chi}^{2}+{\eta}^{2})$ with relations:
\begin{align*}
\begin{split}
[{\eta}^{i}, {\eta}^{j}] = -2{\delta}_{ij}\gamma, \quad [\gamma, \gamma] = [\gamma, {\eta}^{i}] = [\lambda, \gamma] = [\lambda, {\eta}^{i}] = [{\chi}^{i}, \gamma] = [{\chi}^{i}, {\eta}^{j}] = 0.
\end{split}
\end{align*}
We write the $\Lambda$-brackets as:
\begin{align*}
\begin{split}
[a \ _{\Lambda} \ b]
 = \sum\limits_{j \geq 0}\frac{\lambda^{j}}{j!}{a}_{(j|00)}b
 &- {\chi}^{1}\sum\limits_{j \geq 0}\frac{\lambda^{j}}{j!}{a}_{(j|10)}b
\\
 &- {\chi}^{2}\sum\limits_{j \geq 0}\frac{\lambda^{j}}{j!}{a}_{(j|01)}b
 - {\chi}^{1}{\chi}^{2}\sum\limits_{j \geq 0}\frac{\lambda^{j}}{j!}{a}_{(j|11)}b.
\end{split}
\end{align*}
\end{definition}

\begin{definition}\label{definition2.17}{\rm \cite{HK07}}
\rm An $N_{K} = 2$ \textit{supersymmetric vertex algebra} is a tuple $($$V$, $D^{1}$, $D^{2}$, $[$ $_{\Lambda}$ $]$, $\ket{0}$, : :$)$ such that:
\\
$\bullet$ $(V, D^{1}, D^{2}, [ \ _{\Lambda} \ ])$ is an $N_{K}=2$ SUSY Lie conformal algebra,
\\
$\bullet$ $(V, D^{1}, D^{2}, \ket{0},: \ :)$ is a unital differential superalgebra with odd derivations $D^{1}$ and $D^{2}$ satisfying the following properties:
\\
(Quasi-commutativity)
\begin{align*}
\begin{split}
ab - (-1)^{ab}ba = \int_{-\nabla}^{0} [a \ _{\Lambda} \ b] \ d\Lambda,
\end{split}
\end{align*}
(Quasi-associativity)
\begin{align*}
\begin{split}
(ab)c - a(bc) = \left( \int_{0}^{\nabla} \ d\Lambda a \right)[b \ _{\Lambda} \ c] + (-1)^{ab}\left( \int_{0}^{\nabla} \ d\Lambda b \right)[a \ _{\Lambda} \ c],
\end{split}
\end{align*}
$\bullet$ the $\Lambda$-bracket and the product $: \:$ are related by:
\\
(Non-commutative Wick formula)
\begin{align*}
\begin{split}
[a \ _{\Lambda} \ bc] = [a \ _{\Lambda} \ b]c + (-1)^{ab}b[a \ _{\Lambda} \ c] + \int_{0}^{\Lambda} [[a \ _{\Lambda} \ b] \ _{\Gamma} \ c] \ d\Gamma,
\end{split}
\end{align*}
where the integral $\int_{0}^{\Lambda} d\Gamma$ is ${\partial}_{{\eta}^{1}}{\partial}_{{\eta}^{2}}\int_{0}^{\lambda} d\gamma$.
\end{definition}

\section{Superconformal vertex algebras}\label{section3}
\setcounter{equation}{0}

In this section, we review the definitions of $N=1$ and $N=2$ superconformal vertex algebras in the context of the vertex algebras and the SUSY vertex algebras. Details of the definitions can be found in \cite{HK07} and \cite{Kac98}.

\begin{definition}\label{definition3.1}{\rm \cite{Kac98}}
\rm Let $V$ be a vertex algebra. An even vector $L \in V$ is called a \textit{conformal vector} if it satisfies the following conditions:
\\
$\bullet$ $Y(L, z)$ is a Virasoro field with central charge $c$, i.e.
\begin{align*}
\begin{split}
[L \ _\lambda \ L] = (\partial + 2\lambda)L + \frac{c}{12}\lambda^{3},
\end{split}
\end{align*}
$\bullet$ $L_{(0)} = \partial$,
\\
$\bullet$ $L_{(1)}$ is diagonalizable on $V$.
\\
For $v \in V$, we say that $v$ has \textit{conformal weight} $\Delta \in \mathbb{C}$ if:
\begin{align}
\begin{split}
[L \ _\lambda \ v] = (\partial + \Delta\lambda)v + O({\lambda}^{2}),
\end{split}
\end{align}
where $O({\lambda}^{2})$ is a polynomial in $\lambda$ which has no constant and linear terms. Moreover, if there is no $O({\lambda}^{2})$ term in (3.1), the vector $v$ is called \textit{primary}.
\end{definition}

\begin{definition}\label{definition3.2}{\rm \cites{Hel09, Kac98}}
\rm The \textit{$N=1$ superconformal vertex algebra} is generated by a conformal vector $L$, and an odd vector $G$ which is called a \textit{superconformal vector}, satisfying the following super-Virasoro relation:
\begin{align*}
\begin{split}
[L \ _\lambda \ L&] = (\partial + 2\lambda)L + \frac{c}{12}{\lambda}^{3},
\\
[L \ _\lambda& \ G] = (\partial + \frac{3}{2}\lambda)G,
\\
[G \ &_\lambda \ G] = 2L + \frac{c}{3}{\lambda}^{2}.
\end{split}
\end{align*}
\end{definition}

\begin{definition}\label{definition3.3}{\rm \cites{Hel09, Kac98}}
\rm The \textit{$N=2$ superconformal vertex algebra} is generated by a conformal vector $L$, an even vector $J$, and two odd vectors $G^{+}$, $G^{-}$, satisfying the following $\lambda$-bracket relations:
\begin{align*}
\begin{split}
[&L \ _\lambda \ L] = (\partial + 2\lambda)L + \frac{c}{12}{\lambda}^{3}, \quad [L \ _\lambda \ G^{\pm}] = (\partial + \frac{3}{2}\lambda){G}^{\pm},
\\
&[G^{\pm} \ _\lambda \ G^{\pm}] = 0, \quad [G^{+} \ _\lambda \ G^{-}] = L + (\frac{1}{2}\partial + \lambda)J + \frac{c}{6}{\lambda}^{2},
\\
[&L \ _\lambda \ J] = (\partial + \lambda)J, \quad [G^{\pm} \ _\lambda \ J] = {\mp}{G}^{\pm}, \quad [J \ _\lambda \ J] = \frac{c}{3}\lambda.
\end{split}
\end{align*}
\end{definition}

\begin{definition}\label{definition3.4}{\rm \cite{HK07}}
\rm Let $V$ be an $N_{K}=1$ SUSY vertex algebra. A vector $T \in V$ is called an \textit{$N=1$ superconformal vector} if it satisfies the following conditions:
\\
$\bullet$ $Y(T, z, \theta)$ is a super-Virasoro field with central charge $c$, i.e.
\begin{align*}
\begin{split}
[T \ _\Lambda \ T] = (2\partial + 3\lambda + \chi D)T + \frac{c}{3}{\lambda}^{2}{\chi},
\end{split}
\end{align*}
$\bullet$ $T_{(0|0)} = 2\partial$, $T_{(0|1)} = D$,
\\
$\bullet$ $T_{(1|0)}$ is diagonalizable on $V$.
\\
For $v \in V$, we say that $v$ has \textit{conformal weight} $\Delta \in \mathbb{C}$ if:
\begin{align}
\begin{split}
[T \ _\Lambda \ v] = (2\partial + 2\Delta\lambda + \chi D)v + O({\Lambda}^{2}),
\end{split}
\end{align}
where $O({\Lambda}^{2})$ is a polynomial in $\Lambda$ which has no constant and linear terms. Moreover, if there is no $O({\Lambda}^{2})$ term in (3.2), the vector $v$ is called \textit{primary}.
\end{definition}

\begin{remark}\label{remark3.5}{\rm \cites{Hel09, HK07}}
\rm As an $N_{K}=1$ SUSY vertex algebra, the $N=2$ superconformal vertex algebra is generated by an $N=1$ superconformal vector $T$ and an even vector $J$ satisfying the following $N_{K}=1$ $\Lambda$-bracket relations:
\begin{align*}
\begin{split}
[T \ _\Lambda \ T] &= (2\partial + 3\lambda + \chi D)T + \frac{c}{3}{\lambda}^{2}{\chi},
\\
[T \ _\Lambda& \ J] = (2\partial + 2\lambda + \chi D)J,
\\
&[J \ _\Lambda \ J] = T + \frac{c}{3}{\lambda}{\chi}.
\end{split}
\end{align*}
\end{remark}

\begin{definition}\label{definition3.6}{\rm \cite{HK07}}
\rm Let $V$ be an $N_{K}=2$ SUSY vertex algebra. A vector $P \in V$ is called an \textit{$N=2$ superconformal vector} if it satisfies the following conditions:
\\
$\bullet$ $Y(P, z, {\theta}^{1}, {\theta}^{2})$ is an $N=2$ super-Virasoro field with central charge $c$, i.e.
\begin{align*}
\begin{split}
[P \ _\Lambda \ P] = (2\partial + 2\lambda + {\chi}^{1}{D}^{1} + {\chi}^{2}{D}^{2})P + \frac{c}{3}{\lambda}{\chi}^{1}{\chi}^{2},
\end{split}
\end{align*}
$\bullet$ $P_{(0|00)} = 2\partial$, $P_{(0|10)} = - {D}^{1}$, $P_{(0|01)} = {D}^{2}$,
\\
$\bullet$ $P_{(1|00)}$ is diagonalizable on $V$.
\\
For $v \in V$, we say that $v$ has \textit{conformal weight} $\Delta \in \mathbb{C}$ if:
\begin{align}
\begin{split}
[P \ _\Lambda \ v] = (2\partial + 2\Delta\lambda + {\chi}^{1}{D}^{1} + {\chi}^{2}{D}^{2})v + O({\Lambda}^{2}),
\end{split}
\end{align}
where $O({\Lambda}^{2})$ is a polynomial in $\Lambda$ which has no constant and linear terms. Moreover, if there is no $O({\Lambda}^{2})$ term in (3.3), the vector $v$ is called \textit{primary}.
\end{definition}

\section{Superconformal structures of SUSY charged free fermion vertex algebras}\label{section4}
\setcounter{equation}{0}

For the construction of superconformal vectors of the SUSY $W$-algebras in Remark \ref{remark4.10}, a particular deformation of the standard superconformal vector $T^{ch}_{st}$ of the supersymmetric charged free fermion vertex algebra $F^{ch}_{N=1}$ is needed.
In this section, we prove our main result Theorem \ref{theorem4.6} and find the corresponding $N=2$ superconformal symmetry of $F^{ch}_{N=1}$ in Theorem \ref{theorem4.14}. We also construct an $N_{K}=2$ version of the $bc$-$\beta\gamma$ system, so that Theorem \ref{theorem4.6} and Theorem \ref{theorem4.14} can be unified by the $N_{K}=2$ SUSY vertex algebra formalism in Theorem \ref{theorem4.21}.

\subsection{Standard superconformal vectors of SUSY charged free fermion vertex algebras}\label{section4.1}

First, we recall the standard superconformal vectors of the SUSY charged free fermion vertex algebras and the connection with the superconformal vectors of a non-SUSY $bc$-$\beta\gamma$ system. For the SUSY charged free fermion vertex algebra $F^{ch}_{N=1}$ associated to the vector superspace $U$, recall that $A = A_{\bar{0}} \cup A_{\bar{1}}$ refers to the basis for $U$.

\begin{example}\label{example4.1}
\rm In \textup{\cite{KRW03}}, the \textit{standard conformal vector} $L^{ch}_{st}$, for the non-SUSY charged free fermion vertex algebra associated to $\mathfrak{osp}(1|2)$, is defined as
\begin{align*}
\begin{split}
L^{ch}_{st} = -\frac{1}{2}\varphi^{1/2}\partial\varphi_{1/2} + \frac{1}{2}\partial\varphi^{1/2}\varphi_{1/2} + \partial\varphi_{1}\varphi^{1},
\end{split}
\end{align*}
and we can easily find that the corresponding \textit{standard superconformal vector} of $F^{ch}(\mathfrak{osp}(1|2))$ is
\begin{align*}
\begin{split}
G^{ch}_{st} = \varphi^{1/2}\varphi^{1} + \partial\varphi_{1}\varphi_{1/2},
\end{split}
\end{align*}
i.e. the above two vectors satisfy the following $\lambda$-brackets:
\begin{align*}
\begin{split}
[L^{ch}_{st} \ _\lambda \ L^{ch}_{st}] = (\partial + 2\lambda)L^{ch}_{st} + \frac{c}{12}\lambda^{3},
\end{split}
\end{align*}
\begin{align*}
\begin{split}
[L^{ch}_{st} \ _\lambda \ G^{ch}_{st}] = \Big(\partial + \frac{3}{2}\lambda\Big)G^{ch}_{st},
\end{split}
\end{align*}
\begin{align*}
\begin{split}
[G^{ch}_{st} \ _\lambda \ G^{ch}_{st}] = 2L^{ch}_{st} + \frac{c}{3}\lambda^{2},
\end{split}
\end{align*}
with central charge $c=-3$.
If we rename the fields $\varphi_{1/2}$, $\varphi^{1/2}$, $\varphi^{1}$, $\varphi_{1}$, as $b$, $c$, $\beta$, $\gamma$, respectively with reversing the parity, then the vertex algebra $F^{ch}(\mathfrak{osp}(1|2))$ with $L^{ch}_{st}$ and $G^{ch}_{st}$ is nothing but the well-known $bc$-$\beta\gamma$ system with the conformal vector and the superconformal vector:
\begin{align*}
\begin{split}
L^{ch}_{st} = \frac{1}{2}(-c\partial b + \partial cb + 2\partial\gamma\beta), \quad G^{ch}_{st} = c\beta + \partial\gamma b,
\end{split}
\end{align*}
with central charge 3. More detailed explanations of $bc$-$\beta\gamma$ system are in \textup{\cite{EHKZ09}}. For details of the physical statements, the reader is referred to \textup{\cite{Pol98}}.
\end{example}

\begin{definition}\label{definition4.2}
\rm For an $N=1$ superconformal structure of the SUSY charged free fermion vertex algebra $F^{ch}_{N=1}$, we slightly modify the superconformal vector of the SUSY $bc$-$\beta\gamma$ system in Example 5.12 of \cite{HK07}.
Let $T^{ch}_{st}$ be a \textit{standard superconformal vector} of $F^{ch}_{N=1}$ defined by
\begin{align*}
\begin{split}
T^{ch}_{st} = \sum\limits_{a \in A_{\bar{0}}}(\partial\phi_{a}\phi^{\bar{a}} + D\phi_{a}D\phi^{\bar{a}}) + \sum\limits_{a \in A_{\bar{1}}}(\phi_{a}\partial\phi^{\bar{a}} + D\phi_{a}D\phi^{\bar{a}}).
\end{split}
\end{align*}
Then the field $T^{ch}_{st}$ satisfies the $N=1$ superconformal relation
\begin{align*}
\begin{split}
[T^{ch}_{st} \ _{\Lambda} \ T^{ch}_{st}] = (2\partial + 3\lambda + \chi D)T^{ch}_{st} + \frac{c_{st}}{3}\lambda^{2}\chi,
\end{split}
\end{align*}
where the central charge $c_{st}$ is $3 \mathrm{dim}_{\mathbb{C}}U$. The verification of the above $\Lambda$-bracket follows from Theorem \ref{theorem4.6}.
\end{definition}

\begin{remark}\label{remark4.3}{\rm \cites{BZHS08, EHKZ09, EHKZ13, Hel09}}
\rm Let $F^{ch}_{N=1}$ be the SUSY charged free fermion vertex algebra associated to the vector superspace $U = U_{\bar{0}} \oplus U_{\bar{1}}$. For a basis element $a$ in $A_{\bar{0}}$, if we write each superfields $\phi_{a}$, $\phi^{\bar{a}}$ as follows:
\begin{align}
\begin{split}
Y(\phi_{a}, z, \theta) = Y(\gamma,z) + \theta Y(c,z), \quad Y(\phi^{\bar{a}}, z, \theta) = Y(b,z) + \theta Y(\beta,z),
\end{split}
\end{align}
then we can see that the $A_{\bar{0}}$-part of $T^{ch}_{st}$ is the $\mathrm{dim}_{\mathbb{C}}U_{\bar{0}}$ copies of the standard superconformal vector of non-SUSY $bc$-$\beta\gamma$ system, i.e.
\begin{align*}
\begin{split}
Y(\sum\limits_{a \in A_{\bar{0}}}(\partial\phi_{a}\phi^{\bar{a}} + D\phi_{a}D\phi^{\bar{a}}), z, \theta) = Y(\sum\limits_{a \in A_{\bar{0}}}(G^{ch}_{st}),z) + 2\theta Y(\sum\limits_{a \in A_{\bar{0}}}(L^{ch}_{st}),z).
\end{split}
\end{align*}
Similarly, if we expand the superfields of $\phi_{a}$, $\phi^{\bar{a}}$ as
\begin{align*}
\begin{split}
Y(\phi^{\bar{a}}, z, \theta) = Y(\gamma,z) + \theta Y(c,z), \quad Y(\phi_{a}, z, \theta) = Y(b,z) + \theta Y(\beta,z),
\end{split}
\end{align*}
for an odd element $a$ in $A$, then $T^{ch}_{st}$ is the sum of $\mathrm{dim}_{\mathbb{C}}U$ copies of $G^{ch}_{st}$ and $L^{ch}_{st}$, i.e.
\begin{align*}
\begin{split}
T^{ch}_{st} = \sum\limits_{a \in A}G^{ch}_{st} + 2\theta (\sum\limits_{a \in A}L^{ch}_{st}).
\end{split}
\end{align*}
In Remark \ref{remark4.11}, we also get a similar result for our shifted superconformal vectors.
\end{remark}

\begin{proof}
For any even element $a$, the superfield expansion of $Y(D\phi_{a}, z, \theta)$ is obtained as
\begin{align}
\begin{split}
Y(D\phi_{a}, z, \theta) & = (\partial_{\theta} + \theta \partial_{z})Y(\phi_{a}, z, \theta)
\\& = (\partial_{\theta} + \theta \partial_{z})(Y(\gamma,z) + \theta Y(c,z))
\\& = Y(c,z) + \theta Y(\partial\gamma,z)
\end{split}
\end{align}
by using the rules in Theorem 4.16 of \cite{HK07}.
Similarly, we have
\begin{align}
\begin{split}
Y(D\phi^{\bar{a}}, z, \theta) = Y(\beta,z) + \theta Y(\partial b,z),
\end{split}
\end{align}
so that
\begin{align}
\begin{split}
Y(D\phi_{a}D\phi^{\bar{a}}, z, \theta) & = (Y(c,z) + \theta Y(\partial\gamma,z))(Y(\beta,z) + \theta Y(\partial b,z))
\\& = Y(c\beta,z) + \theta Y(\partial\gamma\beta - c\partial b,z).
\end{split}
\end{align}
The negative sign of $c\partial b$ follows from the anti-commutativity of the odd field $Y(c,z)$ and the odd indeterminate $\theta$.
The expansion of $Y(\partial\phi_{a}\phi^{\bar{a}}, z, \theta)$ is also given by
\begin{align}
\begin{split}
Y(\partial\phi_{a}\phi^{\bar{a}}, z, \theta) & = Y(\partial\phi_{a}, z, \theta)Y(\phi^{\bar{a}}, z, \theta)
\\& = \partial_{z}Y(\phi_{a}, z, \theta)Y(\phi^{\bar{a}}, z, \theta)
\\& = (Y(\partial\gamma,z) + \theta Y(\partial c,z))(Y(b,z) + \theta Y(\beta,z))
\\& = Y(\partial\gamma b,z) + \theta Y(\partial cb + \partial\gamma\beta,z),
\end{split}
\end{align}
hence we have
\begin{align}
\begin{split}
Y(\partial\phi_{a}\phi^{\bar{a}} + D\phi_{a}D\phi^{\bar{a}}, z, \theta) & = Y(c\beta + \partial\gamma b,z) + 2\theta Y\Big(\frac{1}{2}(-c\partial b + \partial cb + 2\partial\gamma\beta),z\Big)
\\& = Y(G^{ch}_{st},z) + 2\theta Y(L^{ch}_{st},z).
\end{split}
\end{align}
For an odd element $a$, we also get the following superfield expansion
\begin{align}
\begin{split}
Y(\phi_{a}\partial\phi^{\bar{a}}, z, \theta) & = Y(\phi_{a}, z, \theta)Y(\partial\phi^{\bar{a}}, z, \theta)
\\& = Y(\phi_{a}, z, \theta)\partial_{z}Y(\phi^{\bar{a}}, z, \theta)
\\& = (Y(b,z) + \theta Y(\beta,z))(Y(\partial\gamma,z) + \theta Y(\partial c,z))
\\& = Y(b\partial\gamma,z) + \theta Y(\beta\partial\gamma - b\partial c,z)
\\& = Y(\partial\gamma b,z) + \theta Y(\partial cb + \partial\gamma\beta,z),
\\ Y(D\phi_{a}D\phi^{\bar{a}}, z, \theta) & = (Y(\beta,z) + \theta Y(\partial b,z))(Y(c,z) + \theta Y(\partial\gamma,z))
\\& = Y(\beta c,z) + \theta Y(\partial bc + \beta\partial\gamma,z)
\\& = Y(c\beta,z) + \theta Y(\partial\gamma\beta - c\partial b,z),
\end{split}
\end{align}
hence the field $Y(\phi_{a}\partial\phi^{\bar{a}} + D\phi_{a}D\phi^{\bar{a}}, z, \theta)$ of the $A_{\bar{1}}$-part has the same superfield expansion as in (4.6).
Thus we obtain the result
\begin{align*}
\begin{split}
Y(T^{ch}_{st}, z, \theta) & = Y(\sum\limits_{a \in A}G^{ch}_{st},z) + 2\theta Y(\sum\limits_{a \in A}L^{ch}_{st},z).
\end{split}
\end{align*}
\end{proof}

\subsection{Shifted superconformal vectors of SUSY charged free fermion vertex algebras}\label{section4.2}

In Theorem \ref{theorem4.6}, we define the shifted superconformal vector $T^{ch}_{sh}$, which is a deformation of the standard superconformal vector $T^{ch}_{st}$, by finding the proper ghost term using Ansatz \ref{ansatz4.5} that the requiring shifted superconformal vector is a linear combination of three particular quadratic monomial terms, in order to obtain varying conformal weights of each fields $\phi_{a}$ and $\phi^{\bar{a}}$.

\begin{lemma}\label{lemma4.4} Let $F^{ch}_{N=1}$ be the SUSY charged free fermion vertex algebra. For any even element $a$ in $A$, we have the following $\Lambda$-brackets\textup{:}
\begin{align*}
\begin{split}
[\partial\phi_{a}\phi^{\bar{a}} \ _\Lambda \ \partial\phi_{a}\phi^{\bar{a}}]
& = \partial(\partial\phi_{a}\phi^{\bar{a}}) + 2\lambda(\partial\phi_{a}\phi^{\bar{a}}),
\end{split}
\end{align*}
\begin{align*}
\begin{split}
[\partial\phi_{a}\phi^{\bar{a}} \ _\Lambda \ \phi_{a}\partial\phi^{\bar{a}}]
& = \partial(\phi_{a}\partial\phi^{\bar{a}}) + 2\lambda(\phi_{a}\partial\phi^{\bar{a}}) + \lambda^{2}(\phi_{a}\phi^{\bar{a}}),
\end{split}
\end{align*}
\begin{align}
\begin{split}
[\partial\phi_{a}\phi^{\bar{a}} \ _\Lambda \ D\phi_{a}D\phi^{\bar{a}}] 
& = \partial(D\phi_{a}D\phi^{\bar{a}}) + \lambda(D\phi_{a}D\phi^{\bar{a}}) + \chi D(D\phi_{a}D\phi^{\bar{a}})
\\
& \quad + \frac{1}{2}\lambda^{2}\chi - \lambda\chi(D\phi_{a}\phi^{\bar{a}}),
\end{split}
\end{align}
\begin{align*}
\begin{split}
[\phi_{a}\partial\phi^{\bar{a}} \ _\Lambda \ \partial\phi_{a}\phi^{\bar{a}}]
& = - \partial(\partial\phi_{a}\phi^{\bar{a}}) - 2\lambda(\partial\phi_{a}\phi^{\bar{a}}) - \lambda^{2}(\phi_{a}\phi^{\bar{a}}),
\end{split}
\end{align*}
\begin{align*}
\begin{split}
[\phi_{a}\partial\phi^{\bar{a}} \ _\Lambda \ \phi_{a}\partial\phi^{\bar{a}}]
& = - \partial(\phi_{a}\partial\phi^{\bar{a}}) - 2\lambda(\phi_{a}\partial\phi^{\bar{a}}),
\end{split}
\end{align*}
\begin{align*}
\begin{split}
[\phi_{a}\partial\phi^{\bar{a}} \ _\Lambda \ D\phi_{a}D\phi^{\bar{a}}]
& = - \partial(D\phi_{a}D\phi^{\bar{a}}) - \lambda(D\phi_{a}D\phi^{\bar{a}}) - \chi D(D\phi_{a}D\phi^{\bar{a}})
\\
& \quad + \frac{1}{2}\lambda^{2}\chi - \lambda\chi(\phi_{a}D\phi^{\bar{a}}),
\end{split}
\end{align*}
\begin{align*}
\begin{split}
[D\phi_{a}D\phi^{\bar{a}} \ _\Lambda \ \partial\phi_{a}\phi^{\bar{a}}]
& = \partial(\partial\phi_{a}\phi^{\bar{a}}) + \lambda(\partial\phi_{a}\phi^{\bar{a}}) + \chi D(\partial\phi_{a}\phi^{\bar{a}})
\\
& \quad + \frac{1}{2}\lambda^{2}\chi + \lambda\chi(D\phi_{a}\phi^{\bar{a}}),
\end{split}
\end{align*}
\begin{align*}
\begin{split}
[D\phi_{a}D\phi^{\bar{a}} \ _\Lambda \ \phi_{a}\partial\phi^{\bar{a}}]
& = \partial(\phi_{a}\partial\phi^{\bar{a}}) + \lambda(\phi_{a}\partial\phi^{\bar{a}}) + \chi D(\phi_{a}\partial\phi^{\bar{a}})
\\
& \quad + \frac{1}{2}\lambda^{2}\chi + \lambda\chi(\phi_{a}D\phi^{\bar{a}}),
\end{split}
\end{align*}
\begin{align}
\begin{split}
[D\phi_{a}D\phi^{\bar{a}} \ _\Lambda \ D\phi_{a}D\phi^{\bar{a}}]
& = \partial(D\phi_{a}D\phi^{\bar{a}}) + 2\lambda(D\phi_{a}D\phi^{\bar{a}}).
\end{split}
\end{align}
\end{lemma}

\begin{proof}
For the $\Lambda$-bracket in (4.8), we have:
\begin{align*}
\begin{split}
[\partial\phi_{a}\phi^{\bar{a}} \ _\Lambda \ D\phi_{a}D\phi^{\bar{a}}] 
&= [\partial\phi_{a}\phi^{\bar{a}} \ _\Lambda \ D\phi_{a}]D\phi^{\bar{a}} 
+ D\phi_{a}[\partial\phi_{a}\phi^{\bar{a}} \ _\Lambda \ D\phi^{\bar{a}}] 
\\
& \quad + \int_{0}^{\Lambda} [[\partial\phi_{a}\phi^{\bar{a}} \ _\Lambda \ D\phi_{a}] \ _{\Gamma} \ D\phi^{\bar{a}}] \ d\Gamma
\\
&= ((D + \chi){\partial}{\phi_{a}})D{\phi^{\bar{a}}} 
+ D{\phi_{a}}((D + \chi)({\partial} + {\lambda}){\phi^{\bar{a}}}) 
\\
& \quad + \int_{0}^{\Lambda} [(D + \chi){\partial}{\phi_{a}} \ _{\Gamma} \ D\phi^{\bar{a}}] \ d\Gamma
\\
&= {\partial}D{\phi_{a}}D{\phi^{\bar{a}}} 
+ {\chi}{\partial}{\phi_{a}}D{\phi^{\bar{a}}} 
+ D{\phi_{a}}{\partial}D{\phi^{\bar{a}}} 
+ {\lambda}D{\phi_{a}}D{\phi^{\bar{a}}}
\\
& \quad - {\chi}D{\phi_{a}}{\partial}{\phi^{\bar{a}}} 
- {\lambda}{\chi}D{\phi_{a}}{\phi^{\bar{a}}} 
+ \int_{0}^{\Lambda} [(D + \chi){\partial}{\phi_{a}} \ _{\Gamma} \ D\phi^{\bar{a}}] \ d\Gamma
\\
&= \partial(D\phi_{a}D\phi^{\bar{a}}) + \lambda(D\phi_{a}D\phi^{\bar{a}})
\\
& \quad + \chi D(D\phi_{a}D\phi^{\bar{a}}) - \lambda\chi(D\phi_{a}\phi^{\bar{a}}) + \frac{1}{2}\lambda^{2}\chi,
\end{split}
\end{align*}
from (A.3) and (A.4) in Appendix \ref{sectionA}, and the non-commutative Wick formula of $N_{K}=1$ SUSY LCAs. Indeed, the integral term is computed as follows:
\begin{align*}
\begin{split}
\int_{0}^{\Lambda} [(D + \chi){\partial}{\phi_{a}} \ _{\Gamma} \ D\phi^{\bar{a}}] \ d\Gamma 
&= \int_{0}^{\Lambda} (\eta - \chi)(-\gamma)(-1)(D + \eta) [{\phi_{a}} \ _{\Gamma} \ {\phi^{\bar{a}}}] \ d\Gamma
\\
&= \int_{0}^{\Lambda} (\eta - \chi){\gamma}{\eta} \ d\Gamma
\\
&= \int_{0}^{\Lambda} -{\gamma}^{2} - {\chi}{\gamma}{\eta} \ d\Gamma
\\
&= {\partial}_{\eta}\left( \int_{0}^{\lambda} -{\gamma}^{2} - {\chi}{\gamma}{\eta} \ d\gamma \right)
\\
&= {\partial}_{\eta}\left( - \frac{1}{3}{\lambda}^{3} - \frac{1}{2}{\chi}{\lambda}^{2}{\eta} \right) = \frac{1}{2}{\chi}{\lambda}^{2}.
\end{split}
\end{align*}
Similarly, by using the results in Lemma \ref{lemmaA.1}, Lemma \ref{lemmaA.2} and Lemma \ref{lemmaA.3}, we can compute the other $\Lambda$-brackets.
\end{proof}

Now we will find the shifted superconformal vectors in the SUSY charged free fermion vertex algebras:
\begin{ansatz}\label{ansatz4.5}
\rm Let $T^{ch}_{sh}$ be the field written as a linear combination of three quadratic terms in Lemma \ref{lemma4.4}, i.e.
\begin{align*}
\begin{split}
T^{ch}_{sh} = m_{1}\partial\phi_{a}\phi^{\bar{a}} + m_{2}\phi_{a}\partial\phi^{\bar{a}} + m_{3}D\phi_{a}D\phi^{\bar{a}}.
\end{split}
\end{align*}
Then the $\Lambda$-bracket of $T^{ch}_{sh}$ with itself follows from Lemma \ref{lemma4.4}:
\begin{align*}
\begin{split}
[T^{ch}_{sh} \ _\Lambda \ T^{ch}_{sh}]& = 
\left(
\begin{aligned}
& (m_{1}^{2} - m_{1}m_{2} + m_{1}m_{3})\partial \\
& + (2m_{1}^{2} - 2m_{1}m_{2} + m_{1}m_{3})\lambda
+ (m_{1}m_{3})\chi D \\
\end{aligned}
\right)\partial\phi_{a}\phi^{\bar{a}}
\\& \quad + 
\left(
\begin{aligned}
& (m_{1}m_{2} - m_{2}^{2} + m_{2}m_{3})\partial \\
& + (2m_{1}m_{2} - 2m_{2}^{2} + m_{2}m_{3})\lambda
+ (m_{2}m_{3})\chi D \\
\end{aligned}
\right)\phi_{a}\partial\phi^{\bar{a}}
\\& \quad + 
\left(
\begin{aligned}
& (m_{1}m_{3} - m_{2}m_{3} + m_{3}^{2})\partial \\
& + (m_{1}m_{3} - m_{2}m_{3} + 2m_{3}^{2})\lambda
+ (m_{1}m_{3}-m_{2}m_{3})\chi D \\
\end{aligned}
\right)D\phi_{a}D\phi^{\bar{a}}
\\& \quad + (m_{1}m_{3} + m_{2}m_{3})\lambda^{2}\chi.
\end{split}
\end{align*}
On the other hand,
\begin{align*}
\begin{split}
(2\partial + 3\lambda + \chi D)T^{ch}_{sh} + \frac{c_{sh}}{3}\lambda^{2}\chi
& = \Big((2m_{1})\partial + (3m_{1})\lambda + (m_{1})\chi D\Big)\partial\phi_{a}\phi^{\bar{a}}
\\
& \quad + \Big((2m_{2})\partial + (3m_{2})\lambda + (m_{2})\chi D\Big)\phi_{a}\partial\phi^{\bar{a}}
\\
& \quad+ \Big((2m_{3})\partial + (3m_{3})\lambda + (m_{3})\chi D\Big)D\phi_{a}D\phi^{\bar{a}}
\\
& \quad+ \frac{c_{sh}}{3}\lambda^{2}\chi.
\end{split}
\end{align*}
By comparing the coefficients of each monomial terms, we know that the triple of the form
\begin{align*}
\begin{split}
(m_{1}, m_{2}, m_{3}) = (t_{a} + 1, t_{a}, 1),
\end{split}
\end{align*}
for any complex number $t_{a}$ satisfies the lambda bracket equality:
\begin{align*}
\begin{split}
[T^{ch}_{sh} \ _\Lambda \ T^{ch}_{sh}] = (2\partial + 3\lambda + \chi D)T^{ch}_{sh} + \frac{c_{sh}}{3}\lambda^{2}\chi,
\end{split}
\end{align*}
where the central charge is $c_{sh} = 6t_{a} + 3$. For the odd element $a \in A$, we obtain the result by replacing the fields $\phi_{a}$ and $\phi^{\bar{a}}$ with $\phi^{\bar{a}}$ and $\phi_{a}$ respectively.
\end{ansatz}
\
\\ Hence we have shown the following theorem:
\begin{theorem}\label{theorem4.6}
Let $T^{ch}_{sh}$ be the field defined by
\begin{align*}
\begin{split}
T^{ch}_{sh} 
& = \sum\limits_{a \in A_{\bar{0}}}(t_{a}+1)\partial\phi_{a}\phi^{\bar{a}} + t_{a}\phi_{a}\partial\phi^{\bar{a}} + D\phi_{a}D\phi^{\bar{a}}
\\
& \quad + \sum\limits_{a \in A_{\bar{1}}}t_{a}\partial\phi_{a}\phi^{\bar{a}} + (t_{a}+1)\phi_{a}\partial\phi^{\bar{a}} + D\phi_{a}D\phi^{\bar{a}},
\end{split}
\end{align*}
where $(t_{a})_{a \in A}$ is a sequence of complex numbers. Then the field $T^{ch}_{sh}$ is a superconformal vector of the SUSY charged free fermion vertex algebra $F^{ch}_{N=1}$ with the central charge 
\begin{align*}
\begin{split}
c_{sh} = \sum\limits_{a \in A}(6t_{a} + 3).
\end{split}
\end{align*}
\end{theorem}

\begin{definition}\label{definition4.7}
\rm The field $T^{ch}_{sh}$ defined in Theorem \ref{theorem4.6} is called the \textit{shifted superconformal vector} of the supersymmetric charged free fermion vertex algebra $F^{ch}_{N=1}$.
\end{definition}

\begin{remark}\label{remark4.8}
\rm If we define the ghost field $T^{ch}_{ghost}$ as
\begin{align*}
\begin{split}
T^{ch}_{ghost} = \sum\limits_{a \in A}t_{a}\partial(\phi_{a}\phi^{\bar{a}}),
\end{split}
\end{align*}
then the shifted superconformal vector $T^{ch}_{sh}$ can be written as
\begin{align*}
\begin{split}
T^{ch}_{sh} = T^{ch}_{st} + T^{ch}_{ghost}.
\end{split}
\end{align*}
Moreover, if the complex numbers $t_{a}$ are all zeros, then the standard superconformal vector $T^{ch}_{st}$ is recovered.
\end{remark}
\
\\ Now, the conformal weight of each of fields $\phi_{a}$ and $\phi^{\bar{a}}$ has one degree of freedom.
\begin{proposition}\label{proposition4.9}
For an even element $a$ in $A$, the conformal weights of $\phi_{a}$ and $\phi^{\bar{a}}$ are respectively $-\frac{1}{2}t_{a}$ and $\frac{1}{2}t_{a} + \frac{1}{2}$. For an odd element $a$ in $A$, the conformal weights are $\frac{1}{2}t_{a} + \frac{1}{2}$ and $-\frac{1}{2}t_{a}$ respectively. Moreover, the fields $\phi_{a}$ and $\phi^{\bar{a}}$ are primary with respect to $T^{ch}_{sh}$.
\end{proposition}

\begin{proof}
From the lambda brackets (A.1), (A.5) and (A.7) in Appendix \ref{sectionA}, we have
\begin{align*}
\begin{split}
[(t_{a}+1) & \partial\phi_{a}\phi^{\bar{a}} + t_{a}\phi_{a}\partial\phi^{\bar{a}} + D\phi_{a}D\phi^{\bar{a}} \ _{\Lambda} \ \phi_{a}] 
\\& = (t_{a}+1)[\partial\phi_{a}\phi^{\bar{a}} \ _{\Lambda} \ \phi_{a}] + t_{a}[\phi_{a}\partial\phi^{\bar{a}} \ _{\Lambda} \ \phi_{a}] + [D\phi_{a}D\phi^{\bar{a}} \ _{\Lambda} \ \phi_{a}]
\\& = (t_{a}+1)\partial\phi_{a} - t_{a}(\partial + \lambda)\phi_{a} + (\partial + \chi D)\phi_{a}
\\& = 2\partial\phi_{a} - t_{a}\lambda\phi_{a} + \chi D\phi_{a},
\end{split}
\end{align*}
for any even element $a$, and also we have
\begin{align*}
\begin{split}
[(t_{a}+1)\partial\phi_{a}\phi^{\bar{a}} + t_{a}&\phi_{a}\partial\phi^{\bar{a}} + D\phi_{a}D\phi^{\bar{a}} \ _{\Lambda} \ \phi^{\bar{a}}] 
\\& = (t_{a}+1)[\partial\phi_{a}\phi^{\bar{a}} \ _{\Lambda} \ \phi^{\bar{a}}] + t_{a}[\phi_{a}\partial\phi^{\bar{a}} \ _{\Lambda} \ \phi^{\bar{a}}] + [D\phi_{a}D\phi^{\bar{a}} \ _{\Lambda} \ \phi^{\bar{a}}]
\\& = (t_{a}+1)(\partial + \lambda)\phi^{\bar{a}} - t_{a}\partial\phi^{\bar{a}} + (\partial + \chi D)\phi^{\bar{a}}
\\& = 2\partial\phi^{\bar{a}} + (t_{a} + 1)\lambda\phi^{\bar{a}} + \chi D\phi^{\bar{a}},
\end{split}
\end{align*}
from the lambda brackets in (A.2), (A.6) and (A.8).
Hence, we have the following $\Lambda$-brackets:
\begin{align*}
\begin{split}
[T^{ch}_{sh} \ _{\Lambda} \ \phi_{a}] & = (2\partial - t_{a}\lambda + \chi D)\phi_{a},
\\
[T^{ch}_{sh} \ _{\Lambda} \ \phi^{\bar{a}}] = &\ (2\partial + (t_{a}+1)\lambda + \chi D)\phi^{\bar{a}},
\end{split}
\end{align*}
for any even element $a$, and
\begin{align*}
\begin{split}
[T^{ch}_{sh} \ _{\Lambda} \ \phi_{a}] = &\ (2\partial + (t_{a}+1)\lambda + \chi D)\phi_{a},
\\
[T^{ch}_{sh} \ _{\Lambda} \ \phi^{\bar{a}}] & = (2\partial - t_{a}\lambda + \chi D)\phi^{\bar{a}},
\end{split}
\end{align*}
for any odd element $a$ in $A$.
\end{proof}

\begin{remark}\label{remark4.10}
\rm Roughly speaking, the SUSY $W$-algebra is the BRST cohomology of the tensor product of the SUSY affine vertex algebra and the SUSY charged free fermion vetex algebra, i.e.:
\begin{align*}
\begin{split}
W(\overline{\mathfrak{g}}, f_{odd}, k) = H(V^{k}(\overline{\mathfrak{g}}) \otimes F^{ch}_{N=1}, Q),
\end{split}
\end{align*}
where $Q$ is the differential of $V^{k}(\overline{\mathfrak{g}}) \otimes F^{ch}_{N=1}$ and $\overline{\mathfrak{g}}$ denotes ${\mathfrak{g}}$ with reversed parity (for details, see Section 4.1 of \cite{MRS21}). For the SUSY affine vertex algebra $V^{k}(\overline{\mathfrak{g}})$, there exists an $N=1$ superconformal vector $T^{KT}_{st}$, the Kac-Todorov superconformal vector which originally came from \cite{KT85}, of the SUSY affine vertex algebra (see Example 5.9 of \cite{HK07}). Also, for any vector $\overline{v} \in \overline{\mathfrak{g}}$, we know that the vector
\begin{align*}
\begin{split}
T^{KT}_{sh} = T^{KT}_{st} + {\partial}\overline{v},
\end{split}
\end{align*}
is also an $N=1$ superconformal vector (see Example 2.13 of \cite{HZ10}). Hence, we propose that the sum of our shifted superconformal vector in Theorem \ref{theorem4.6} and the Kac-Todorov superconformal vector
\begin{align*}
\begin{split}
T^{W}:= T^{KT}_{sh} \otimes 1 + 1 \otimes T^{ch}_{sh}
\end{split}
\end{align*}
gives a superconformal vector of the supersymmetric $W$-algebra.
\end{remark}

\begin{remark}\label{remark4.11}
\rm Decomposition of the field $T^{ch}_{sh}$ via non-SUSY superconformal vectors can be done in the same way as Remark \ref{remark4.3}, i.e. the shifted superconformal field $T^{ch}_{sh}$ of the SUSY charged free fermion vertex algebra $F^{ch}_{N=1}$ is the sum of shifted superconformal fields of the non-SUSY $bc$-$\beta\gamma$ system.
\end{remark}

\begin{proof}
For $a \in A_{\bar{0}}$, we have
\begin{align*}
\begin{split}
Y(\phi_{a}\partial\phi^{\bar{a}}, z, \theta) 
& = Y(\phi_{a}, z, \theta)Y(\partial\phi^{\bar{a}}, z, \theta)
\\& = Y(\phi_{a}, z, \theta)\partial_{z}Y(\phi^{\bar{a}}, z, \theta)
\\& = (Y(\gamma,z) + \theta Y(c,z))(Y(\partial b,z) + \theta Y(\partial\beta,z))
\\& = Y(\gamma\partial b,z) + \theta Y(c\partial b + \gamma\partial\beta,z),
\end{split}
\end{align*}
and, with the equations (4.4) and (4.5), the $A_{\bar{0}}$-part of the shifted superconformal field $T^{ch}_{sh}$ can be expressed as
\begin{align*}
\begin{split}
(t_{a}+1)\partial\phi_{a}&\phi^{\bar{a}} + t_{a}\phi_{a}\partial\phi^{\bar{a}} + D\phi_{a}D\phi^{\bar{a}}
\\
&= (t_{a})\gamma\partial b + (t_{a}+1)\partial\gamma b + c\beta 
\\
& \quad + 2\theta \Big(\frac{1}{2}((t_{a}-1)c\partial b + (t_{a}+1)\partial cb + t_{a}\gamma\partial\beta + (t_{a}+2)\partial\gamma\beta)\Big).
\end{split}
\end{align*}
If we define the ghost parts as
\begin{align*}
\begin{split}
L^{ch}_{ghost} = \frac{t_{a}}{2}(c\partial b + \partial cb + \gamma\partial\beta + \partial\gamma\beta), \quad G^{ch}_{ghost} = t_{a}(\gamma\partial b + \partial\gamma b),
\end{split}
\end{align*}
we obtain the shifted superconformal vectors in the $bc$-$\beta\gamma$ system as
\begin{align}
\begin{split}
L^{ch}_{sh} & = L^{ch}_{st} + L^{ch}_{ghost}
\\ & = \frac{1}{2}((t_{a}-1)c\partial b + (t_{a}+1)\partial cb + t_{a}\gamma\partial\beta + (t_{a}+2)\partial\gamma\beta),
\end{split}
\end{align}
\begin{align}
\begin{split}
G^{ch}_{sh} & = G^{ch}_{st} + G^{ch}_{ghost}
\\& = (t_{a})\gamma\partial b + (t_{a}+1)\partial\gamma b + c\beta.
\end{split}
\end{align}
On the other hand, for the odd elements $a \in A_{\bar{1}}$, we have
\begin{align*}
\begin{split}
Y(\partial\phi_{a}\phi^{\bar{a}}, z, \theta) 
& = Y(\partial\phi_{a}, z, \theta)Y(\phi^{\bar{a}}, z, \theta)
\\& = \partial_{z}Y(\phi_{a}, z, \theta)Y(\phi^{\bar{a}}, z, \theta)
\\& = (Y(\partial b,z) + \theta Y(\partial\beta,z))(Y(\gamma,z) + \theta Y(c,z))
\\& = Y(\partial b\gamma,z) + \theta Y(\partial\beta\gamma - \partial bc, z),
\end{split}
\end{align*}
and, also with the equations in (4.7), the $A_{\bar{1}}$-part of the shifted superconformal field $T^{ch}_{sh}$ can be written as
\begin{align*}
\begin{split}
(t_{a}+1)\partial\phi_{a}&\phi^{\bar{a}} + t_{a}\phi_{a}\partial\phi^{\bar{a}} + D\phi_{a}D\phi^{\bar{a}}
\\
&= (t_{a})\gamma\partial b + (t_{a}+1)\partial\gamma b + c\beta 
\\
& \quad + 2\theta \Big(\frac{1}{2}((t_{a}-1)c\partial b + (t_{a}+1)\partial cb + t_{a}\gamma\partial\beta + (t_{a}+2)\partial\gamma\beta)\Big),
\end{split}
\end{align*}
of which the expression is same as the $A_{\bar{0}}$-parts.
\ 
\\ Hence the shifted superconformal field of the SUSY charged free fermion vertex algebra is the sum of $\mathrm{dim}_{\mathbb{C}}U$ copies of the shifted superconformal field of the non-SUSY $bc$-$\beta\gamma$ system, i.e.
\begin{align*}
\begin{split}
T^{ch}_{sh} = \sum\limits_{a \in A}G^{ch}_{sh} + 2\theta (\sum\limits_{a \in A}L^{ch}_{sh}).
\end{split}
\end{align*}
\end{proof}

\begin{remark}\label{remark4.12}
\rm In Remark \ref{remark4.8}, the SUSY vertex algebra case, the field $T^{ch}_{st}$ is a special case of the field $T^{ch}_{sh}$ when $t_{a} = 0$ for all $a \in A$. Likewise, we can see that the each field of $L^{ch}_{st}$ and $G^{ch}_{st}$ is also a special case of the each field of $L^{ch}_{sh}$ and $G^{ch}_{sh}$ in the non-SUSY $bc$-$\beta\gamma$ system.
\end{remark}

\subsection{Shifted $N=2$ superconformal structures of SUSY charged free fermion vertex algebras} \label{section4.3}

From \cite{BZHS08} and \cite{HK07}, we can use the fact that the field
\begin{align*}
\begin{split}
J^{ch}_{st} = \sum\limits_{a \in A_{\bar{0}}}D\phi_{a}\phi^{\bar{a}} + \sum\limits_{a \in A_{\bar{1}}}D\phi^{\bar{a}}\phi_{a},
\end{split}
\end{align*}
together with $T^{ch}_{st}$ form an $N=2$ superconformal vertex algebra structure.
For the $N=2$ superconformal symmetry of $F^{ch}_{N=1}$ compatible with the shifted superconformal vector $T^{ch}_{sh}$, we find the corresponding ghost term in Theorem \ref{theorem4.14}. We also write the shifted $N=2$ superconformal fields as $N=2$ superfields in Theorem \ref{theorem4.21}.

\begin{lemma}\label{lemma4.13} Let $F^{ch}_{N=1}$ be the SUSY charged free fermion vertex algebra. For any even element $a$ in $A$, we have the following $\Lambda$-brackets\textup{:}
\begin{align}
\begin{split}
[\partial\phi_{a}\phi^{\bar{a}} \ _\Lambda \ D\phi_{a}\phi^{\bar{a}}]
& = \partial(D\phi_{a}\phi^{\bar{a}}) + \lambda(D\phi_{a}\phi^{\bar{a}}) + \chi D(D\phi_{a}\phi^{\bar{a}}) 
\\
& \quad + \chi(D\phi_{a}D\phi^{\bar{a}}) - \frac{1}{2}\lambda^{2},
\end{split}
\end{align}
\begin{align}
\begin{split}
[\partial\phi_{a}\phi^{\bar{a}} \ _\Lambda \ \phi_{a}D\phi^{\bar{a}}]
& = \partial(\phi_{a}D\phi^{\bar{a}}) + \lambda(\phi_{a}D\phi^{\bar{a}}) + \chi D(\phi_{a}D\phi^{\bar{a}}) 
\\
& \quad - \chi(D\phi_{a}D\phi^{\bar{a}}) + \lambda\chi(\phi_{a}\phi^{\bar{a}}) + \frac{1}{2}\lambda^{2},
\end{split}
\end{align}
\begin{align}
\begin{split}
[\phi_{a}\partial\phi^{\bar{a}} \ _\Lambda \ D\phi_{a}\phi^{\bar{a}}]
& = - \partial(D\phi_{a}\phi^{\bar{a}}) - \lambda(D\phi_{a}\phi^{\bar{a}}) - \chi D(D\phi_{a}\phi^{\bar{a}}) 
\\
& \quad - \chi(D\phi_{a}D\phi^{\bar{a}}) - \lambda\chi(\phi_{a}\phi^{\bar{a}}) - \frac{1}{2}\lambda^{2},
\end{split}
\end{align}
\begin{align}
\begin{split}
[\phi_{a}\partial\phi^{\bar{a}} \ _\Lambda \ \phi_{a}D\phi^{\bar{a}}]
& = - \partial(\phi_{a}D\phi^{\bar{a}}) - \lambda(\phi_{a}D\phi^{\bar{a}}) - \chi D(\phi_{a}D\phi^{\bar{a}}) 
\\
& \quad + \chi(D\phi_{a}D\phi^{\bar{a}}) + \frac{1}{2}\lambda^{2},
\end{split}
\end{align}
\begin{align}
\begin{split}
[D\phi_{a}D\phi^{\bar{a}} \ _\Lambda \ D\phi_{a}\phi^{\bar{a}}]
& = \partial(D\phi_{a}\phi^{\bar{a}}) + \lambda(D\phi_{a}\phi^{\bar{a}}) - \chi(D\phi_{a}D\phi^{\bar{a}}) + \frac{1}{2}\lambda^{2},
\end{split}
\end{align}
\begin{align}
\begin{split}
[D\phi_{a}D\phi^{\bar{a}} \ _\Lambda \ \phi_{a}D\phi^{\bar{a}}]
& = \partial(\phi_{a}D\phi^{\bar{a}}) + \lambda(\phi_{a}D\phi^{\bar{a}}) + \chi(D\phi_{a}D\phi^{\bar{a}}) + \frac{1}{2}\lambda^{2},
\end{split}
\end{align}
\begin{align}
\begin{split}
[D\phi_{a}\phi^{\bar{a}} \ _\Lambda \ D\phi_{a}\phi^{\bar{a}}]
& = \partial\phi_{a}\phi^{\bar{a}} + D\phi_{a}D\phi^{\bar{a}} + \lambda\chi,
\end{split}
\end{align}
\begin{align}
\begin{split}
[D\phi_{a}\phi^{\bar{a}} \ _\Lambda \ \phi_{a}D\phi^{\bar{a}}]
& = \phi_{a}\partial\phi^{\bar{a}} - D\phi_{a}D\phi^{\bar{a}} + \lambda(\phi_{a}\phi^{\bar{a}}),
\end{split}
\end{align}
\begin{align}
\begin{split}
[\phi_{a}D\phi^{\bar{a}} \ _\Lambda \ D\phi_{a}\phi^{\bar{a}}]
& = - \partial\phi_{a}\phi^{\bar{a}} - D\phi_{a}D\phi^{\bar{a}} - \lambda(\phi_{a}\phi^{\bar{a}}),
\end{split}
\end{align}
\begin{align}
\begin{split}
[\phi_{a}D\phi^{\bar{a}} \ _\Lambda \ \phi_{a}D\phi^{\bar{a}}]
& = - \phi_{a}\partial\phi^{\bar{a}} + D\phi_{a}D\phi^{\bar{a}} - \lambda\chi.
\end{split}
\end{align}
\end{lemma}

\begin{proof}
For the $\Lambda$-bracket in (4.12), we have:
\begin{align*}
\begin{split}
[\partial\phi_{a}\phi^{\bar{a}} \ _\Lambda \ D\phi_{a}\phi^{\bar{a}}] 
&= [\partial\phi_{a}\phi^{\bar{a}} \ _\Lambda \ D\phi_{a}]\phi^{\bar{a}} 
+ D\phi_{a}[\partial\phi_{a}\phi^{\bar{a}} \ _\Lambda \ \phi^{\bar{a}}] 
\\
& \qquad \qquad \qquad \qquad \quad + \int_{0}^{\Lambda} [[\partial\phi_{a}\phi^{\bar{a}} \ _\Lambda \ D\phi_{a}] \ _{\Gamma} \ \phi^{\bar{a}}] \ d\Gamma
\\
&= ((D + \chi){\partial}{\phi_{a}}){\phi^{\bar{a}}} 
+ D{\phi_{a}}({\partial} + {\lambda}){\phi^{\bar{a}}} 
\\
& \qquad \qquad \qquad \qquad \quad + \int_{0}^{\Lambda} [(D + \chi){\partial}{\phi_{a}} \ _{\Gamma} \ \phi^{\bar{a}}] \ d\Gamma
\\
&= {\partial}D{\phi_{a}}{\phi^{\bar{a}}} 
+ {\chi}{\partial}{\phi_{a}}{\phi^{\bar{a}}} 
+ D{\phi_{a}}{\partial}{\phi^{\bar{a}}} 
+ {\lambda}D{\phi_{a}}{\phi^{\bar{a}}}
\\
& \qquad \qquad \qquad \qquad \quad + \int_{0}^{\Lambda} - (\eta - \chi){\gamma} \ d\Gamma
\\
&= \partial(D\phi_{a}\phi^{\bar{a}}) + \lambda(D\phi_{a}\phi^{\bar{a}}) + \chi D(D\phi_{a}\phi^{\bar{a}}) 
\\
& \qquad \qquad \qquad \qquad \quad + \chi(D\phi_{a}D\phi^{\bar{a}}) - \frac{1}{2}\lambda^{2},
\end{split}
\end{align*}
from (A.3) and (A.2) in Lemma \ref{lemmaA.1}. Similarly, we can compute the $\Lambda$-brackets in (4.13)-(4.17) by using Lemma \ref{lemmaA.1}, Lemma \ref{lemmaA.2} and Lemma \ref{lemmaA.3}. The $\Lambda$-brackets in (4.18)-(4.21) follows from Lemma \ref{lemmaA.4}.
\end{proof}

\begin{theorem}\label{theorem4.14}
Let $J^{ch}_{ghost}$ be the field defined as
\begin{align*}
\begin{split}
J^{ch}_{ghost} = \sum\limits_{a \in A_{\bar{0}}}t_{a}(D\phi_{a}\phi^{\bar{a}}+\phi_{a}D\phi^{\bar{a}}) + \sum\limits_{a \in A_{\bar{1}}}t_{a}(D\phi^{\bar{a}}\phi_{a}+\phi^{\bar{a}}D\phi_{a}).
\end{split}
\end{align*}
Then the current field
\begin{align*}
\begin{split}
J^{ch}_{sh}=J^{ch}_{st}+J^{ch}_{ghost},
\end{split}
\end{align*}
and the superconformal field $T^{ch}_{sh}$ generates an $N=2$ superconformal structure of central charge $c_{sh}$.
\end{theorem}

\begin{proof}
For the $A_{\bar{0}}$-parts, by using the Lemma \ref{lemma4.13}, we have the following $\Lambda$-brackets:
\begin{align*}
\begin{split}
[T^{ch}_{st} \ _{\Lambda} \ J^{ch}_{st}] & = (2\partial + 2\lambda + \chi D)J^{ch}_{st},
\\
[T^{ch}_{st} \ _{\Lambda} \ J^{ch}_{ghost}] & = (2\partial + 2\lambda + \chi D)J^{ch}_{ghost} + \sum\limits_{a \in A_{\bar{0}}}t_{a}\lambda^{2} + \sum\limits_{a \in A_{\bar{0}}}t_{a}\lambda\chi\phi_{a}\phi^{\bar{a}},
\\
[T^{ch}_{ghost} \ _{\Lambda} \ J^{ch}_{st}] & = - \sum\limits_{a \in A_{\bar{0}}}t_{a}\lambda^{2} - \sum\limits_{a \in A_{\bar{0}}}t_{a}\lambda\chi\phi_{a}\phi^{\bar{a}},
\\
[T^{ch}_{ghost} \ _{\Lambda} \ J^{ch}_{ghost}] & = 0,
\end{split}
\end{align*}
hence, we have
\begin{align*}
\begin{split}
[T^{ch}_{sh} \ _{\Lambda} \ J^{ch}_{sh}] & = (2\partial + 2\lambda + \chi D)J^{ch}_{sh}.
\end{split}
\end{align*}
Also, we can check the following $\Lambda$-brackets:
\begin{align*}
\begin{split}
[J^{ch}_{st} \ _{\Lambda} \ J^{ch}_{st}] & = T^{ch}_{st} + \sum\limits_{a \in A_{\bar{0}}}\lambda\chi,
\\
[J^{ch}_{st} \ _{\Lambda} \ J^{ch}_{ghost}] & = T^{ch}_{ghost} + \sum\limits_{a \in A_{\bar{0}}}t_{a}\lambda\phi_{a}\phi^{\bar{a}} + \sum\limits_{a \in A_{\bar{0}}}t_{a}\lambda\chi,
\\
[J^{ch}_{ghost} \ _{\Lambda} \ J^{ch}_{st}] & = - \sum\limits_{a \in A_{\bar{0}}}t_{a}\lambda\phi_{a}\phi^{\bar{a}} + \sum\limits_{a \in A_{\bar{0}}}t_{a}\lambda\chi,
\\
[J^{ch}_{ghost} \ _{\Lambda} \ J^{ch}_{ghost}] & = 0,
\end{split}
\end{align*}
thus we obtain the following result:
\begin{align*}
\begin{split}
[J^{ch}_{sh} \ _{\Lambda} \ J^{ch}_{sh}] & = T^{ch}_{sh} + \sum\limits_{a \in A_{\bar{0}}}(2t_{a} + 1)\lambda\chi.
\end{split}
\end{align*}
For the $A_{\bar{1}}$-part, we have the same result by replacing the fields $\phi_{a}$ with $\phi^{\bar{a}}$, and vice versa.
Hence, the fields $T^{ch}_{sh}$ and $J^{ch}_{sh}$ satisfy the $N=2$ superconformal symmetry.
\end{proof}

\begin{remark}\label{remark4.15}
\rm Indeed, we can write (see Example 3.13 in \cite{BZHS08}) the $N=1$ superfields $T^{ch}_{sh}$ and $J^{ch}_{sh}$ as:
\begin{align}
\begin{split}
T^{ch}_{sh}(z, \theta) & = (G^{ch, +}_{sh}(z) + G^{ch, -}_{sh}(z)) + 2\theta L^{ch}_{sh}(z),
\\
J^{ch}_{sh}(z, \theta) & = J^{ch}_{sh}(z) + \theta(G^{ch, -}_{sh}(z) - G^{ch, +}_{sh}(z)),
\end{split}
\end{align}
where
\begin{align*}
\begin{split}
L^{ch}_{sh}(z) = \sum\limits_{a \in A}\frac{1}{2}
\left(
\begin{aligned}
& (t_{a}-1)c\partial b + (t_{a}+1)\partial cb \\
& + t_{a}\gamma\partial\beta + (t_{a}+2)\partial\gamma\beta \\
\end{aligned}
\right),
\end{split}
\end{align*}
\begin{align*}
\begin{split}
G^{ch, +}_{sh}(z) = \sum\limits_{a \in A}c\beta,
\end{split}
\end{align*}
\begin{align*}
\begin{split}
G^{ch, -}_{sh}(z) = \sum\limits_{a \in A}(t_{a}+1)\partial\gamma b + (t_{a})\gamma\partial b,
\end{split}
\end{align*}
\begin{align*}
\begin{split}
J^{ch}_{sh}(z) = \sum\limits_{a \in A}(t_{a}+1)cb + t_{a}\gamma\beta,
\end{split}
\end{align*}
from the equations (4.1), (4.2), (4.3), (4.10) and (4.11). Then we know that the fields $L^{ch}_{sh}$, $G^{ch, +}_{sh}$, $G^{ch, -}_{sh}$ and $J^{ch}_{sh}$ satisfy the non-SUSY lambda bracket relations of $N=2$ superconformal vertex algebras.
\end{remark}

\begin{remark}\label{remark4.16}
\rm From Definition 4.10 of \cite{HK07}, recall that the lambda brackets of $N_{K}=2$ SUSY Lie conformal algebras are given by:
\begin{align*}
\begin{split}
[u \ _{\Lambda} \ v]
 = \sum\limits_{j \geq 0}\frac{\lambda^{j}}{j!}{u}_{(j|00)}v
 &- {\chi}^{1}\sum\limits_{j \geq 0}\frac{\lambda^{j}}{j!}{u}_{(j|10)}v
\\
 &- {\chi}^{2}\sum\limits_{j \geq 0}\frac{\lambda^{j}}{j!}{u}_{(j|01)}v
 - {\chi}^{1}{\chi}^{2}\sum\limits_{j \geq 0}\frac{\lambda^{j}}{j!}{u}_{(j|11)}v.
\end{split}
\end{align*}
On the other hand, by sesquilinearity of $N_{K}=2$ SUSY LCAs, we have:
\begin{align*}
\begin{split}
{u}_{(j|00)}v &= - D^{1}D^{2}{u}_{(j|11)}v,
\\
{u}_{(j|10)}v &= D^{2}{u}_{(j|11)}v,
\\
{u}_{(j|01)}v &= - D^{1}{u}_{(j|11)}v,
\end{split}
\end{align*}
hence we obtain the following equality:
\begin{align}
\begin{split}
[u \ _{\Lambda} \ v] =
 - [D^{1}D^{2}u \ _{\lambda} \ v]
 &- {\chi}^{1}[D^{2}u \ _{\lambda} \ v]
\\
 &+ {\chi}^{2}[D^{1}u \ _{\lambda} \ v]
 - {\chi}^{1}{\chi}^{2}[u \ _{\lambda} \ v].
\end{split}
\end{align}
\end{remark}

\begin{remark}\label{remark4.17}
\rm Following the Remark \ref{remark4.15}, consider the $N_{K}=2$ SUSY vertex algebra structure (see Remark 2.8 in \cite{HK07}) where the odd endomorphisms are defined by
\begin{align*}
\begin{split}
D^{1} = (G^{ch, +}_{sh} + G^{ch, -}_{sh})_{(0)}& = (G^{ch, +}_{st} + G^{ch, -}_{st})_{(0)},
\\
D^{2} = i(G^{ch, +}_{sh} - G^{ch, -}_{sh})_{(0)}& = i(G^{ch, +}_{st} - G^{ch, -}_{st})_{(0)},
\end{split}
\end{align*}
and the state-field correspondence is given by
\begin{align*}
\begin{split}
Y(v, z, {\theta}^{1}, {\theta}^{2}) = Y(v, z) + {\theta}^{1}Y(D^{1}v, z) + {\theta}^{2}Y(D^{2}v, z) + {\theta}^{2}{\theta}^{1}Y(D^{1}D^{2}v, z).
\end{split}
\end{align*}
Let $P^{ch}_{sh}$ be the vector defined as:
\begin{align*}
\begin{split}
P^{ch}_{sh} & = -i{J^{ch}_{sh}}_{(-1)}\ket{0}
\\
& = -i\sum\limits_{a \in A}(t_{a}+1)cb + t_{a}\gamma\beta.
\end{split}
\end{align*}
Then the $N=2$ superfield corresponding to $P^{ch}_{sh}$ is expanded as:
\begin{align}
\begin{split}
Y(P^{ch}_{sh}, z, {\theta}^{1}, {\theta}^{2})
= &- iY(J^{ch}_{sh}, z) - i{\theta}^{1}Y(D^{1}J^{ch}_{sh}, z)
\\
&- i{\theta}^{2}Y(D^{2}J^{ch}_{sh}, z) - i{\theta}^{2}{\theta}^{1}Y(D^{1}D^{2}J^{ch}_{sh}, z)
\\
= &- iY(J^{ch}_{sh}, z) - i{\theta}^{1}Y(G^{ch, -}_{sh} - G^{ch, +}_{sh}, z)
\\
&- {\theta}^{2}Y(G^{ch, +}_{sh} + G^{ch, -}_{sh}, z) - {\theta}^{2}{\theta}^{1}Y(2L^{ch}_{sh}, z)
\\
= &- iY(J^{ch}_{sh}, z, {\theta}^{1}) - {\theta}^{2}Y(T^{ch}_{sh}, z, {\theta}^{1}),
\end{split}
\end{align}
which follows from the non-SUSY lambda bracket relations of $N=2$ superconformal vertex algebras.
\end{remark}

\begin{proposition}\label{proposition4.18}
The field $P^{ch}_{sh}$ defined in Remark \ref{remark4.17} satisfies the following lambda bracket:
\begin{align*}
\begin{split}
[P^{ch}_{sh} \ _{\Lambda} \ P^{ch}_{sh}]_{N_{K}=2}
= (2\partial + 2\lambda + {\chi}^{1}D^{1} + {\chi}^{2}D^{2})P^{ch}_{sh} 
+ \frac{c_{sh}}{3}\lambda{\chi}^{1}{\chi}^{2}.
\end{split}
\end{align*}
Hence, the shifted $N=2$ superconformal symmetry in Theorem \ref{theorem4.14} is generated by one $N=2$ superfield $P^{ch}_{sh}$ in the $N_{K}=2$ SUSY vertex algebra formalism.
\end{proposition}

\begin{proof}
Note that the fields $L^{ch}_{sh}$, $G^{ch, +}_{sh}$, $G^{ch, -}_{sh}$, $J^{ch}_{sh}$ in Remark \ref{remark4.15} satisfy the non-SUSY lambda bracket relations of $N=2$ superconformal vertex algebras (see Definition \ref{definition3.3}), and by the definitions of $D^{1}$ and $D^{2}$ in Remark \ref{remark4.17}, we have:
\begin{align*}
\begin{split}
{D}^{1}{J}^{ch}_{sh} &= G^{ch, -}_{sh} - G^{ch, +}_{sh},
\\
i{D}^{2}{J}^{ch}_{sh} &= G^{ch, +}_{sh} + G^{ch, -}_{sh},
\\
i{D}^{1}{D}^{2}{J}^{ch}_{sh} &= 2L^{ch}_{sh}.
\end{split}
\end{align*}
Then, from the equality (4.23) in Remark \ref{remark4.16}, we have:
\begin{align*}
\begin{split}
[P^{ch}_{sh} \ _{\Lambda} \ P^{ch}_{sh}]
= &- i[i{D}^{1}{D}^{2}{J}^{ch}_{sh} \ _{\lambda} \ {J}^{ch}_{sh}]
 - i{\chi}^{1}[i{D}^{2}{J}^{ch}_{sh} \ _{\lambda} \ {J}^{ch}_{sh}]
\\
 & \qquad \qquad \quad
 - {\chi}^{2}[{D}^{1}{J}^{ch}_{sh} \ _{\lambda} \ {J}^{ch}_{sh}]
 + {\chi}^{1}{\chi}^{2}[{J}^{ch}_{sh} \ _{\lambda} \ {J}^{ch}_{sh}]
\\
= &- i[2L^{ch}_{sh} \ _{\lambda} \ {J}^{ch}_{sh}]
 - i{\chi}^{1}[G^{ch, +}_{sh} + G^{ch, -}_{sh} \ _{\lambda} \ {J}^{ch}_{sh}]
\\
 & \qquad \qquad \quad
 - {\chi}^{2}[G^{ch, -}_{sh} - G^{ch, +}_{sh} \ _{\lambda} \ {J}^{ch}_{sh}]
 + {\chi}^{1}{\chi}^{2}[{J}^{ch}_{sh} \ _{\lambda} \ {J}^{ch}_{sh}]
\\
= &- i(2\partial + 2\lambda){J}^{ch}_{sh}
 - i{\chi}^{1}(G^{ch, -}_{sh} - G^{ch, +}_{sh})
\\
 & \qquad \qquad \quad
 - {\chi}^{2}(G^{ch, +}_{sh} + G^{ch, -}_{sh})
 + {\chi}^{1}{\chi}^{2}\frac{c_{sh}}{3}\lambda
\\
= &- i(2\partial + 2\lambda){J}^{ch}_{sh}
 - i{\chi}^{1}{D}^{1}{J}^{ch}_{sh} - i{\chi}^{2}{D}^{2}{J}^{ch}_{sh}
 + {\chi}^{1}{\chi}^{2}\frac{c_{sh}}{3}\lambda
\\
= & \ (2\partial + 2\lambda + {\chi}^{1}D^{1} + {\chi}^{2}D^{2})(-i{J}^{ch}_{sh})
+ \frac{c_{sh}}{3}\lambda{\chi}^{1}{\chi}^{2}.
\end{split}
\end{align*}
\end{proof}

\begin{remark}\label{remark4.19}
\rm Now, using the notations in Remark \ref{remark4.3}, let us define the $N=2$ superfields ${\Phi}_{a}$ and ${\Phi}^{\bar{a}}$ as:
\begin{align*}
\begin{split}
{\Phi}_{a} 
&= Y(\gamma, z, {\theta}^{1}, {\theta}^{2}) 
\\
&= Y(\gamma, z) + {\theta}^{1}Y({D}^{1}\gamma, z) + {\theta}^{2}Y({D}^{2}\gamma, z) + {\theta}^{2}{\theta}^{1}Y({D}^{1}{D}^{2}\gamma, z)
\\
&= Y(\gamma, z) + {\theta}^{1}Y(c, z) + {\theta}^{2}Y(ic, z) + {\theta}^{2}{\theta}^{1}Y(i\partial\gamma, z)
\\
&= {\phi}_{a} + i{\theta}^{2}{D}^{1}{\phi}_{a},
\end{split}
\end{align*}
\begin{align*}
\begin{split}
{\Phi}^{\bar{a}}
&= Y(b, z, {\theta}^{1}, {\theta}^{2}) 
\\
&= Y(b, z) + {\theta}^{1}Y({D}^{1}b, z) + {\theta}^{2}Y({D}^{2}b, z) + {\theta}^{2}{\theta}^{1}Y({D}^{1}{D}^{2}b, z)
\\
&= Y(b, z) + {\theta}^{1}Y(\beta, z) + {\theta}^{2}Y(i\beta, z) + {\theta}^{2}{\theta}^{1}Y(i\partial b, z)
\\
&= {\phi}^{\bar{a}} + i{\theta}^{2}{D}^{1}{\phi}^{\bar{a}},
\end{split}
\end{align*}
where the $N_{K}=2$ SUSY vertex algebra structure is given by Remark \ref{remark4.17}. Then the $\Lambda$-bracket of ${\Phi}_{a}$ and ${\Phi}^{\bar{a}}$ is:
\begin{align}
\begin{split}
[{\Phi}_{a} \ _{\Lambda} \ {\Phi}^{\bar{a}}] = - i{\chi}^{1} + {\chi}^{2}.
\end{split}
\end{align}
We also know that:
\begin{align*}
\begin{split}
{D}^{1}{\Phi}_{a} = {D}^{1}{\phi}_{a} - i{\theta}^{2}\partial{\phi}_{a},
\quad & 
{D}^{1}{\Phi}^{\bar{a}} = {D}^{1}{\phi}^{\bar{a}} - i{\theta}^{2}\partial{\phi}^{\bar{a}},
\\
{D}^{2}{\Phi}_{a} = i{D}^{1}{\phi}_{a} + {\theta}^{2}\partial{\phi}_{a},
\quad & 
{D}^{2}{\Phi}^{\bar{a}} = i{D}^{1}{\phi}^{\bar{a}} + {\theta}^{2}\partial{\phi}^{\bar{a}},
\end{split}
\end{align*}
so that, for any even element $a \in A_{\bar{0}}$, we have:
\begin{align*}
\begin{split}
{D}^{1}{\Phi}_{a}{\Phi}^{\bar{a}} &=
{D}^{1}{\phi}_{a}{\phi}^{\bar{a}} - i{\theta}^{2}(\partial{\phi}_{a}{\phi}^{\bar{a}} + {D}^{1}{\phi}_{a}{D}^{1}{\phi}^{\bar{a}}),
\\
{\Phi}_{a}{D}^{1}{\Phi}^{\bar{a}} &=
{\phi}_{a}{D}^{1}{\phi}^{\bar{a}} - i{\theta}^{2}({\phi}_{a}\partial{\phi}^{\bar{a}} - {D}^{1}{\phi}_{a}{D}^{1}{\phi}^{\bar{a}}),
\end{split}
\end{align*}
and, for any odd element $a \in A_{\bar{1}}$:
\begin{align*}
\begin{split}
{D}^{1}{\Phi}^{\bar{a}}{\Phi}_{a} &=
{D}^{1}{\phi}^{\bar{a}}{\phi}_{a} - i{\theta}^{2}({\phi}_{a}\partial{\phi}^{\bar{a}} + {D}^{1}{\phi}_{a}{D}^{1}{\phi}^{\bar{a}}),
\\
{\Phi}^{\bar{a}}{D}^{1}{\Phi}_{a} &=
{\phi}^{\bar{a}}{D}^{1}{\phi}_{a} - i{\theta}^{2}(\partial{\phi}_{a}{\phi}^{\bar{a}} - {D}^{1}{\phi}_{a}{D}^{1}{\phi}^{\bar{a}}).
\end{split}
\end{align*}
Therefore, the $N=2$ superfield $P^{ch}_{sh}$ in (4.24) can be written as:
\begin{align*}
\begin{split}
P^{ch}_{sh}
= - i\left(
\begin{aligned}
& \sum\limits_{a \in A_{\bar{0}}}(t_{a}+1){D}^{1}{\Phi}_{a}{\Phi}^{\bar{a}} + t_{a}{\Phi}_{a}{D}^{1}{\Phi}^{\bar{a}}
\\
& \quad + \sum\limits_{a \in A_{\bar{1}}}(t_{a}+1){D}^{1}{\Phi}^{\bar{a}}{\Phi}_{a} + t_{a}{\Phi}^{\bar{a}}{D}^{1}{\Phi}_{a}
\\
\end{aligned}
\right).
\end{split}
\end{align*}
Then, by using the $N_{K}=2$ lambda bracket (4.25) and the relations:
\begin{align}
\begin{split}
i{D}^{1}{\Phi}_{a} = {D}^{2}{\Phi}_{a}, \quad i{D}^{1}{\Phi}^{\bar{a}} = {D}^{2}{\Phi}^{\bar{a}},
\end{split}
\end{align}
one can also check directly the field $P^{ch}_{sh}$ is an $N=2$ superconformal vector (see Theorem \ref{theorem4.21}).
\end{remark}

\begin{definition}\label{definition4.20}
\rm Let $U$ be a finite-dimensional vector superspace and $A = A_{\bar{0}} \cup A_{\bar{1}}$ be a basis of $U$, where $A_{\bar{0}}$ is the even part and $A_{\bar{1}}$ is the odd part of the basis. Define two vector superspaces:
\begin{align*}
\begin{split}
\Phi_{U} \simeq U, \quad \Phi^{\bar{U}} \simeq \Pi{U}^{*},
\end{split}
\end{align*}
whose basis elements are denoted by $\Phi_{a}$ and $\Phi^{\bar{a}}$ respectively. Define the $\Lambda$-brackets on $R^{bc\beta\gamma}_{N=2} = {\mathcal{H}}_{N=2} \otimes ({\Phi}_{U} \oplus {\Phi}^{\bar{U}})$ by:
\begin{align*}
\begin{split}
[{\Phi}_{a} \ _{\Lambda} \ {\Phi}^{\bar{b}}] = {\delta}_{ab}(- i{\chi}^{1} + {\chi}^{2}), \quad [{\Phi}_{a} \ _{\Lambda} \ {\Phi}_{b}] = [{\Phi}^{\bar{a}} \ _{\Lambda} \ {\Phi}^{\bar{b}}] = 0,
\end{split}
\end{align*}
for $a, b \in A$. Then the ${\mathcal{H}}_{N=2}$-module $R^{bc\beta\gamma}_{N=2}$ is an $N_{K} = 2$ SUSY Lie conformal algebra. Now let $V(R^{bc\beta\gamma}_{N=2})$ be the universal enveloping $N_{K} = 2$ SUSY vertex algebra of $R^{bc\beta\gamma}_{N=2}$, and let $I(R^{bc\beta\gamma}_{N=2})$ be the ideal generated by elements of the form (4.26). Then we define the \textit{$N=2$ supersymmetric $bc$-$\beta\gamma$ system} associated to the vector superspace $U$ as the quotient $V(R^{bc\beta\gamma}_{N=2})/I(R^{bc\beta\gamma}_{N=2})$.
\end{definition}

\begin{theorem}\label{theorem4.21}
Let $P^{ch}_{sh}$ be the vector defined by
\begin{align*}
\begin{split}
P^{ch}_{sh} = - i\left(
\begin{aligned}
& \sum\limits_{a \in A_{\bar{0}}}(t_{a}+1){D}^{1}{\Phi}_{a}{\Phi}^{\bar{a}} + t_{a}{\Phi}_{a}{D}^{1}{\Phi}^{\bar{a}}
\\
& \quad + \sum\limits_{a \in A_{\bar{1}}}(t_{a}+1){D}^{1}{\Phi}^{\bar{a}}{\Phi}_{a} + t_{a}{\Phi}^{\bar{a}}{D}^{1}{\Phi}_{a}
\\
\end{aligned}
\right),
\end{split}
\end{align*}
where $(t_{a})_{a \in A}$ is a sequence of complex numbers. Then the vector $P^{ch}_{sh}$ is an $N=2$ superconformal vector of the $N=2$ SUSY $bc$-$\beta\gamma$ system with the central charge 
\begin{align*}
\begin{split}
c_{sh} = \sum\limits_{a \in A}(6t_{a} + 3).
\end{split}
\end{align*}
\end{theorem}

\begin{proof}
Let $a$ be an even element of $A$. From the non-commutative Wick formula of $N_{K}=2$ SUSY LCAs, we have:
\begin{align*}
\begin{split}
[{\Phi}_{a} \ _{\Lambda} \ {D}^{1}{\Phi}_{a}{\Phi}^{\bar{a}}] = (i{\chi}^{1}{D}^{1} - {\chi}^{2}{D}^{1}){\Phi}_{a},
\end{split}
\end{align*}
hence, by skew-symmetry:
\begin{align}
\begin{split}
[{D}^{1}{\Phi}_{a}{\Phi}^{\bar{a}} \ _{\Lambda} \ {\Phi}_{a}] &= - [{\Phi}_{a} \ _{- \nabla - \Lambda} \ {D}^{1}{\Phi}_{a}{\Phi}^{\bar{a}}]
\\
&= (-i(-{D}^{1}-{\chi}^{1}){D}^{1} + (-{D}^{2}-{\chi}^{2}){D}^{1}){\Phi}_{a}
\\
&= i(2{\partial} + {\chi}^{1}{D}^{1} + {\chi}^{2}{D}^{2}){\Phi}_{a}.
\end{split}
\end{align}
Therefore we have:
\begin{align}
\begin{split}
[{D}^{1}{\Phi}_{a}{\Phi}^{\bar{a}} \ _{\Lambda} \ {D}^{1}{\Phi}_{a}] 
&= i({D}^{1}+{\chi}^{1})(2{\partial} + {\chi}^{1}{D}^{1} + {\chi}^{2}{D}^{2}){\Phi}_{a}
\\
&= i(2{\partial} + {\lambda} + {\chi}^{1}{D}^{1} + {\chi}^{2}{D}^{2}){D}^{1}{\Phi}_{a} - {\chi}^{1}{\chi}^{2}{D}^{1}{\Phi}_{a}.
\end{split}
\end{align}
Similarly, we obtain:
\begin{align}
\begin{split}
[{D}^{1}{\Phi}_{a}{\Phi}^{\bar{a}} \ _{\Lambda} \ {\Phi}^{\bar{a}}] = i(2{\partial} + {\lambda} + {\chi}^{1}{D}^{1} + {\chi}^{2}{D}^{2}){\Phi}^{\bar{a}} + {\chi}^{1}{\chi}^{2}{\Phi}^{\bar{a}}.
\end{split}
\end{align}
By the $\Lambda$-brackets in (4.28) and (4.29), and from the non-commutative Wick formula, we have:
\begin{align}
\begin{split}
[{D}^{1}{\Phi}_{a}{\Phi}^{\bar{a}} \ _{\Lambda} \ {D}^{1}{\Phi}_{a}{\Phi}^{\bar{a}}] 
&= [{D}^{1}{\Phi}_{a}{\Phi}^{\bar{a}} \ _{\Lambda} \ {D}^{1}{\Phi}_{a}]{\Phi}^{\bar{a}} + {D}^{1}{\Phi}_{a}[{D}^{1}{\Phi}_{a}{\Phi}^{\bar{a}} \ _{\Lambda} \ {\Phi}^{\bar{a}}]
\\
& \quad + \int_{0}^{\Lambda} [[{D}^{1}{\Phi}_{a}{\Phi}^{\bar{a}} \ _{\Lambda} \ {D}^{1}{\Phi}_{a}] \ _{\Gamma} \ {\Phi}^{\bar{a}}] \ d\Gamma
\\
&= i(2{\partial} + 2{\lambda} + {\chi}^{1}{D}^{1} + {\chi}^{2}{D}^{2})({D}^{1}{\Phi}_{a}{\Phi}^{\bar{a}}) - {\lambda}{\chi}^{1}{\chi}^{2}.
\end{split}
\end{align}
On the other hand, we can compute:
\begin{align}
\begin{split}
[{\Phi}_{a}{\Phi}^{\bar{a}} \ _{\Lambda} \ {\Phi}_{a}{\Phi}^{\bar{a}}] = 0.
\end{split}
\end{align}
From the $\Lambda$-brackets in (4.27) and (4.29), we have:
\begin{align*}
\begin{split}
[{D}^{1}{\Phi}_{a}{\Phi}^{\bar{a}} \ _{\Lambda} \ {\Phi}_{a}{\Phi}^{\bar{a}}] &= i(2{\partial} + {\lambda} + {\chi}^{1}{D}^{1} + {\chi}^{2}{D}^{2})({\Phi}_{a}{\Phi}^{\bar{a}}) 
\\
& \quad + {\chi}^{1}{\chi}^{2}({\Phi}_{a}{\Phi}^{\bar{a}}) - {\lambda}(-i{\chi}^{1} + {\chi}^{2}).
\end{split}
\end{align*}
Hence, using sesquilinearity, we find:
\begin{align}
\begin{split}
[{D}^{1}{\Phi}_{a}{\Phi}^{\bar{a}} \ _{\Lambda} \ {D}^{1}({\Phi}_{a}{\Phi}^{\bar{a}})] &= i(2{\partial} + 2{\lambda} + {\chi}^{1}{D}^{1} + {\chi}^{2}{D}^{2})({D}^{1}({\Phi}_{a}{\Phi}^{\bar{a}}))
\\
& \quad + (i{\lambda}{\chi}^{1} + {\lambda}{\chi}^{2})({\Phi}_{a}{\Phi}^{\bar{a}}) - i{\lambda}^{2} - {\lambda}{\chi}^{1}{\chi}^{2},
\end{split}
\end{align}
and, by skew-symmetry, we have:
\begin{align}
\begin{split}
[{D}^{1}({\Phi}_{a}{\Phi}^{\bar{a}}) \ _{\Lambda} \ {D}^{1}{\Phi}_{a}{\Phi}^{\bar{a}}] = -(i{\lambda}{\chi}^{1} + {\lambda}{\chi}^{2})({\Phi}_{a}{\Phi}^{\bar{a}}) + i{\lambda}^{2} - {\lambda}{\chi}^{1}{\chi}^{2}.
\end{split}
\end{align}
Combining the results in (4.30), (4.31), (4.32) and (4.33), we obtain:
\begin{align*}
\begin{split}
[{D}^{1}{\Phi}_{a}&{\Phi}^{\bar{a}} + t_{a}{D}^{1}({\Phi}_{a}{\Phi}^{\bar{a}}) \ _{\Lambda} \ {D}^{1}{\Phi}_{a}{\Phi}^{\bar{a}} + t_{a}{D}^{1}({\Phi}_{a}{\Phi}^{\bar{a}})]
\\
&= i(2{\partial} + 2{\lambda} + {\chi}^{1}{D}^{1} + {\chi}^{2}{D}^{2})({D}^{1}{\Phi}_{a}{\Phi}^{\bar{a}} + t_{a}{D}^{1}({\Phi}_{a}{\Phi}^{\bar{a}})) - (2{t_{a}}+1){\lambda}{\chi}^{1}{\chi}^{2},
\end{split}
\end{align*}
hence, we get the result.
\end{proof}


\newpage
\begin{appendices}
\section{Calculations of lambda brackets}\label{sectionA}
\setcounter{equation}{0}

For any even element $a$ in $A$, we have the following lambda brackets in $F^{ch}_{N=1}$.

\begin{lemma}\label{lemmaA.1} Lambda brackets for $\partial\phi_{a}\phi^{\bar{a}}$\textup{:}
\begin{align}
\begin{split}
[\partial\phi_{a}\phi^{\bar{a}} \ _\Lambda \ \phi_{a}]
& = \partial\phi_{a},
\end{split}
\end{align}
\begin{align}
\begin{split}
[\partial\phi_{a}\phi^{\bar{a}} \ _\Lambda \ \phi^{\bar{a}}]
& = (\partial + \lambda)\phi^{\bar{a}},
\end{split}
\end{align}
\begin{align}
\begin{split}
[\partial\phi_{a}\phi^{\bar{a}} \ _\Lambda \ D\phi_{a}]
& = (D + \chi)\partial\phi_{a},
\end{split}
\end{align}
\begin{align}
\begin{split}
[\partial\phi_{a}\phi^{\bar{a}} \ _\Lambda \ D\phi^{\bar{a}}]
& = (D + \chi)(\partial + \lambda)\phi^{\bar{a}},
\end{split}
\end{align}
\begin{align*}
\begin{split}
[\partial\phi_{a}\phi^{\bar{a}} \ _\Lambda \ \partial\phi_{a}]
& = (\partial + \lambda)\partial\phi_{a},
\end{split}
\end{align*}
\begin{align*}
\begin{split}
[\partial\phi_{a}\phi^{\bar{a}} \ _\Lambda \ \partial\phi^{\bar{a}}]
& = (\partial + \lambda)^{2}\phi^{\bar{a}}.
\end{split}
\end{align*}
\end{lemma}

\begin{proof}
The $\Lambda$-brackets can be obtained by using the rules of $N_{K}=1$ SUSY vertex algebras. For the $\Lambda$-bracket in (A.1):
\begin{align*}
\begin{split}
[\phi_{a} \ _\Lambda \ \partial\phi_{a}\phi^{\bar{a}}]
& = [\phi_{a} \ _\Lambda \ \partial\phi_{a}]\phi^{\bar{a}} + (-1)^{(a+1)a}{\partial}{\phi_{a}}[\phi_{a} \ _\Lambda \ \phi^{\bar{a}}] + \int_{0}^{\Lambda} [[\phi_{a} \ _\Lambda \ \partial\phi_{a}] \ _{\Gamma} \ {\phi}^{\bar{a}}] \ d\Gamma
\\
&= \partial\phi_{a},
\end{split}
\end{align*}
since, from sesquilinearity, we know that:
\begin{align*}
\begin{split}
[\phi_{a} \ _\Lambda \ \partial\phi_{a}] = (\partial + \lambda)[\phi_{a} \ _\Lambda \ \phi_{a}] = 0.
\end{split}
\end{align*}
Hence, by skew-symmetry, we get:
\begin{align*}
\begin{split}
[\partial\phi_{a}\phi^{\bar{a}} \ _\Lambda \ \phi_{a}] = (-1)^{(a+a+1)a}[\phi_{a} \ _{- \nabla - \Lambda} \ \partial\phi_{a}\phi^{\bar{a}}] = \partial\phi_{a}.
\end{split}
\end{align*}
Similarly, the other $\Lambda$-brackets also can be obtained by using the rules of $\Lambda$-brackets of $N_{K}=1$ SUSY vertex algebras.
\end{proof}

\begin{lemma}\label{lemmaA.2} Lambda brackets for $\phi_{a}\partial\phi^{\bar{a}}$\textup{:}
\begin{align}
\begin{split}
[\phi_{a}\partial\phi^{\bar{a}} \ _\Lambda \ \phi_{a}]
& = - (\partial + \lambda)\phi_{a},
\end{split}
\end{align}
\begin{align}
\begin{split}
[\phi_{a}\partial\phi^{\bar{a}} \ _\Lambda \ \phi^{\bar{a}}]
& = - \partial\phi^{\bar{a}},
\end{split}
\end{align}
\begin{align*}
\begin{split}
[\phi_{a}\partial\phi^{\bar{a}} \ _\Lambda \ D\phi_{a}]
& = - (D + \chi)(\partial + \lambda)\phi_{a},
\end{split}
\end{align*}
\begin{align*}
\begin{split}
[\phi_{a}\partial\phi^{\bar{a}} \ _\Lambda \ D\phi^{\bar{a}}]
& = - (D + \chi)\partial\phi^{\bar{a}},
\end{split}
\end{align*}
\begin{align*}
\begin{split}
[\phi_{a}\partial\phi^{\bar{a}} \ _\Lambda \ \partial\phi_{a}]
& = - (\partial + \lambda)^{2}\phi_{a},
\end{split}
\end{align*}
\begin{align*}
\begin{split}
[\phi_{a}\partial\phi^{\bar{a}} \ _\Lambda \ \partial\phi^{\bar{a}}]
& = -(\partial + \lambda)\partial\phi^{\bar{a}}.
\end{split}
\end{align*}
\end{lemma}

\begin{proof}
It follows from direct calculations.
\end{proof}

\begin{lemma}\label{lemmaA.3} Lambda brackets for $D\phi_{a}D\phi^{\bar{a}}$\textup{:}
\begin{align}
\begin{split}
[D\phi_{a}D\phi^{\bar{a}} \ _\Lambda \ \phi_{a}]
& = (\partial + \chi D)\phi_{a},
\end{split}
\end{align}
\begin{align}
\begin{split}
[D\phi_{a}D\phi^{\bar{a}} \ _\Lambda \ \phi^{\bar{a}}]
& = (\partial + \chi D)\phi^{\bar{a}},
\end{split}
\end{align}
\begin{align*}
\begin{split}
[D\phi_{a}D\phi^{\bar{a}} \ _\Lambda \ D\phi_{a}]
& = (\partial + \lambda)D\phi_{a},
\end{split}
\end{align*}
\begin{align*}
\begin{split}
[D\phi_{a}D\phi^{\bar{a}} \ _\Lambda \ D\phi^{\bar{a}}]
& = (\partial + \lambda)D\phi^{\bar{a}},
\end{split}
\end{align*}
\begin{align*}
\begin{split}
[D\phi_{a}D\phi^{\bar{a}} \ _\Lambda \ \partial\phi_{a}]
& = (\partial + \lambda)(\partial + \chi D)\phi_{a},
\end{split}
\end{align*}
\begin{align*}
\begin{split}
[D\phi_{a}D\phi^{\bar{a}} \ _\Lambda \ \partial\phi^{\bar{a}}]
& = (\partial + \lambda)(\partial + \chi D)\phi^{\bar{a}}.
\end{split}
\end{align*}
\end{lemma}

\begin{proof}
The results can be verified by direct computations.
\end{proof}

\begin{lemma}\label{lemmaA.4} Lambda brackets for $D\phi_{a}\phi^{\bar{a}}$ and $\phi_{a}D\phi^{\bar{a}}$\textup{:}
\begin{align*}
\begin{split}
[D\phi_{a}\phi^{\bar{a}} \ _\Lambda \ \phi_{a}]
& = - D\phi_{a},
\end{split}
\end{align*}
\begin{align*}
\begin{split}
[D\phi_{a}\phi^{\bar{a}} \ _\Lambda \ \phi^{\bar{a}}]
& = - (D + \chi)\phi^{\bar{a}},
\end{split}
\end{align*}
\begin{align*}
\begin{split}
[D\phi_{a}\phi^{\bar{a}} \ _\Lambda \ D\phi_{a}]
& = (\partial + \chi D)\phi_{a},
\end{split}
\end{align*}
\begin{align*}
\begin{split}
[D\phi_{a}\phi^{\bar{a}} \ _\Lambda \ D\phi^{\bar{a}}]
& = (\partial + \lambda)\phi^{\bar{a}},
\end{split}
\end{align*}
\begin{align*}
\begin{split}
[\phi_{a}D\phi^{\bar{a}} \ _\Lambda \ \phi_{a}]
& = (D + \chi)\phi_{a},
\end{split}
\end{align*}
\begin{align*}
\begin{split}
[\phi_{a}D\phi^{\bar{a}} \ _\Lambda \ \phi^{\bar{a}}]
& = D\phi^{\bar{a}},
\end{split}
\end{align*}
\begin{align*}
\begin{split}
[\phi_{a}D\phi^{\bar{a}} \ _\Lambda \ D\phi_{a}]
& = - (\partial + \lambda)\phi_{a},
\end{split}
\end{align*}
\begin{align*}
\begin{split}
[\phi_{a}D\phi^{\bar{a}} \ _\Lambda \ D\phi^{\bar{a}}]
& = - (\partial + \chi D)\phi^{\bar{a}}.
\end{split}
\end{align*}
\end{lemma}

\begin{proof}
We also get the results from direct computations.
\end{proof}


\section{Charge decomposition}\label{sectionB}
\setcounter{equation}{0}

\begin{definition}\label{definitionB.1}
\rm As in Section 5.1 of \textup{\cite{HZ10}}, the eigenvalue of ${J^{ch}_{sh}}_{(0|1)}$ corresponding to $v \in F^{ch}_{N=1}$ is called the \textit{charge} of $v$. Then, for any even element $a$ in $A$, the charges of $\phi_{a}$ and $\phi^{\bar{a}}$ are $t_{a}$ and $-(t_{a} + 1)$ respectively. For the odd elements, the charges of $\phi_{a}$ and $\phi^{\bar{a}}$ are $-(t_{a} + 1)$ and $t_{a}$ respectively. Also, if we define the \textit{BRST operator} $Q$ of $F^{ch}_{N=1}$ as follows:
\begin{align*}
\begin{split}
Q = \frac{1}{2}\left( {T^{ch}_{sh}}_{(0|1)} - {J^{ch}_{sh}}_{(0|0)} \right),
\end{split}
\end{align*}
then we obtain the following result.
\end{definition}

\begin{proposition}\label{propositionB.2}
The BRST operator $Q$ of $F^{ch}_{N=1}$ satisfies $Q^{2} = 0$.
\end{proposition}

\begin{proof}
If we set $d = \sum\limits_{a \in A}{D\phi_{a}D\phi^{\bar{a}}}$, we know that:
\begin{align*}
\begin{split}
Q & = \frac{1}{2}\left( {T^{ch}_{sh}}_{(0|1)} - {J^{ch}_{sh}}_{(0|0)} \right)
\\
& = \frac{1}{2}\left( {T^{ch}_{sh}}_{(0|1)} - {DJ^{ch}_{sh}}_{(0|1)} \right) = d_{(0|1)},
\end{split}
\end{align*}
by sesquilinearity of $N_{K}=1$ supersymmetric LCAs. Also we have:
\begin{align}
\begin{split}
[d \ _{\Lambda} \ [d \ _{\Gamma} \ v]] 
& = [[d \ _{\Lambda} \ d] \ _{\Lambda + \Gamma} \ v] + [d \ _{\Gamma} \ [d \ _{\Lambda} \ v]]
\\
& = [\partial d + 2\lambda d \ _{\Lambda + \Gamma} \ v] + [d \ _{\Gamma} \ [d \ _{\Lambda} \ v]]
\\
& = (\lambda - \gamma)[d \ _{\Lambda + \Gamma} \ v] + [d \ _{\Gamma} \ [d \ _{\Lambda} \ v]],
\end{split}
\end{align}
for $v \in F^{ch}_{N=1}$, from the Jacobi identity of $N_{K}=1$ supersymmetric LCAs and the equation (4.9) in Lemma \ref{lemma4.4}. Hence, by comparing the $\chi\eta$-terms in (B.1), we have the result:
\begin{align*}
\begin{split}
d_{(0|1)}(d_{(0|1)}v) = 0.
\end{split}
\end{align*}
\end{proof}

\begin{remark}\label{remarkB.3}
\rm The BRST operator $Q$ is invariant under any values of $t_{a}$, that is:
\begin{align*}
\begin{split}
Q = \frac{1}{2}\left( {T^{ch}_{st}}_{(0|1)} - {J^{ch}_{st}}_{(0|0)} \right).
\end{split}
\end{align*}
Indeed, if we decompose the fields $T^{ch}_{sh}$ and $J^{ch}_{sh}$ as in (4.22), then the BRST operator $Q$ is the zero mode of the field $G^{ch, +}_{sh}$. On the other hand, the ghost term contributes to the \textit{homotopy operator} $H$ of $F^{ch}_{N=1}$ defined by:
\begin{align*}
\begin{split}
H = \frac{1}{2}\left( {T^{ch}_{sh}}_{(0|1)} + {J^{ch}_{sh}}_{(0|0)} \right) = ( T^{ch}_{sh} - d )_{(0|1)},
\end{split}
\end{align*}
which is the zero mode of the field $G^{ch, -}_{sh}$.
\end{remark}

\begin{proposition}\label{propositionB.4}
Let  $v \in F^{ch}_{N=1}$ be a vector of the charge $m$. Then $Q(v)$ has the charge $m+1$.
\end{proposition}

\begin{proof}
From the Lemma \ref{lemma4.13}, we have:
\begin{align*}
\begin{split}
[d \ _{\Lambda} \ J^{ch}_{sh}] = (\partial + \lambda)J^{ch}_{sh} - \chi d +  \frac{c_{sh}}{6}{\lambda}^{2},
\end{split}
\end{align*}
where $d = \sum\limits_{a \in A}{D\phi_{a}D\phi^{\bar{a}}}$. Hence we have the following $\Lambda$-bracket:
\begin{align*}
\begin{split}
[J^{ch}_{sh} \ _{\Lambda} \ d] & = [d \ _{-\Lambda-\nabla} \ J^{ch}_{sh}]
\\
& = - \lambda J^{ch}_{sh} + (\chi + D)d + \frac{c_{sh}}{6}{\lambda}^{2},
\end{split}
\end{align*}
by skew-symmetry of $N_{K}=1$ supersymmetric LCAs, and then:
\begin{align}
\begin{split}
[J^{ch}_{sh} \ _{\Lambda} \ [d \ _{\Gamma} \ v]] 
& = - [[J^{ch}_{sh} \ _{\Lambda} \ d] \ _{\Lambda + \Gamma} \ v] + [d \ _{\Gamma} \ [J^{ch}_{sh} \ _{\Lambda} \ v]]
\\
& = [\lambda J^{ch}_{sh} - (\chi + D)d - \frac{c_{sh}}{6} \ _{\Lambda + \Gamma} \ v] + [d \ _{\Gamma} \ [J^{ch}_{sh} \ _{\Lambda} \ v]]
\\
& = \lambda[J^{ch}_{sh} \ _{\Lambda + \Gamma} \ v] - \eta[d \ _{\Lambda + \Gamma} \ v] + [d \ _{\Gamma} \ [J^{ch}_{sh} \ _{\Lambda} \ v]],
\end{split}
\end{align}
from the Jacobi identity of $N_{K}=1$ supersymmetric LCAs. By comparing the coefficients of $\chi\eta$-terms in (B.2), we obtain the result:
\begin{align*}
\begin{split}
{J^{ch}_{sh}}_{(0|1)}Q(v) & = {J^{ch}_{sh}}_{(0|1)}(d_{(0|1)}v) 
\\
& = d_{(0|1)}v + d_{(0|1)}({J^{ch}_{sh}}_{(0|1)}v)
\\
& = Q(v) + Q(mv) = (m+1)Q(v).
\end{split}
\end{align*}
\end{proof}

\begin{remark}\label{remarkB.5}
\rm For the BRST operator $Q$ and the homotopy operator $H$ of $F^{ch}_{N=1}$, defined in Remark \ref{remarkB.3}, we have the following result from the simple calculations:
\begin{align*}
\begin{split}
Q(D^{i}\phi) =
\begin{cases}
D^{i+1}\phi , & \text{if} \ i \in 2\mathbb{Z}_{\geq 0}
\\
0 , & \text{if} \ i \in 2\mathbb{Z}_{\geq 0}+1,
\end{cases}
\\
H(D^{i}\phi) =
\begin{cases}
0 , & \text{if} \ i \in 2\mathbb{Z}_{\geq 0}
\\
D^{i+1}\phi , & \text{if} \ i \in 2\mathbb{Z}_{\geq 0}+1,
\end{cases}
\end{split}
\end{align*}
where $\phi$ stands either $\phi_{a}$ or $\phi^{\bar{a}}$ for any $a \in A$. Also the homotopy operator $H$ decreases the charge by $-1$. Moreover, if the charges of $u, v \in F^{ch}_{N=1}$ are $m_{u}$ and $m_{v}$ respectively, then $uv$ has the charge $m_{u} + (-1)^{u}m_{v}$.
\end{remark}

\begin{proposition}\label{propositionB.6}
For $(t_{a})_{a \in A} \in \mathbb{Z}^{\mathrm{dim}_{\mathbb{C}}U}$, the supersymmetric charged free fermion vertex algebra has a $\mathbb{Z}$-grading\textup{:}
\begin{align*}
\begin{split}
F^{ch}_{N=1} = \bigoplus_{m \in \mathbb{Z}}F^{ch, m}_{N=1},
\end{split}
\end{align*}
called the charge decomposition, where $F^{ch, m}_{N=1}$ is the subspace whose elements have the charge $m$.
\end{proposition}

\begin{proof}
Combining the facts in Remark \ref{remarkB.5} and Proposition \ref{propositionB.4}, we obtain the result.
\end{proof}

\end{appendices}


\end{document}